\documentclass{article}
\usepackage{amsmath}

\usepackage{float}
\usepackage{graphicx}
\usepackage{enumerate}
\usepackage{epstopdf}
\usepackage{color}
\usepackage{bm}
\usepackage{cases}
\usepackage{multirow}
\usepackage{hhline}
\usepackage{vcell}
\usepackage{amssymb}
\usepackage{fullpage}

\newcommand{\bn}{{\mathbf n}}

\newcommand{\dx}{\mathrm{d}x}

\newcommand{\dt}{\mathrm{d}t}

\newcommand{\drho}{\mathrm{d}\rho}

\newcommand{\bxi}{\mathbf{\xi}}
\newcommand{\bhatn}{\hat{\mathbf{n}}}

\newtheorem{theorem}{\bf Theorem}[section]

\newtheorem{remark}{\bf Remark}[section]

\newtheorem{proof}{Proof}[section]
\title{A new Phase-Field Model for Anisotropic Surface Diffusion: Anisotropic Cahn-Hilliard Equation with Improved Conservation (ACH-IC)}

\author{
Zeyu Zhou\thanks{Department of Mathematics, Southern University of Science and Technology, Shenzhen 518055, China (\texttt{zhouzy2021@mail.sustech.edu.cn}).}		
\and Wei Jiang\thanks{School of Mathematics and Statistics, Hubei Key Laboratory of Computational Science, Wuhan University, Wuhan 430072, China
(\texttt{jiangwei1007@whu.edu.cn}).}
\and Tiezheng Qian\thanks{Department of Mathematics, The Hong Kong University of Science and Technology, Hong Kong SAR 999077, China	
	(\texttt{maqian@ust.hk}).}
\and Zhen Zhang\thanks{Corresponding author. Department of Mathematics,  National Center for Applied Mathematics (Shenzhen), Southern University of Science and Technology, Shenzhen 518055, China
(\texttt{zhangz@sustech.edu.cn}).}
}

\begin{document}
\maketitle

%%%% Abstract text to be placed here %%%%%%%%%%%%
\begin{abstract}
As popular approximations to sharp-interface models, the Cahn-Hilliard type phase-field models are usually used to simulate interface dynamics with volume conservation. However, the convergence rate of the volume enclosed by the interface to its sharp-interface limit is usually at first order of the interface thickness in the classical Cahn-Hilliard model with constant or degenerate mobilities. In this work, we propose a variational framework for developing new Cahn-Hilliard dynamics with enhanced volume conservation by introducing a more general conserved quantity. In particular, based on Onsager's variational principle (OVP) and a modified conservation law, we develop an anisotropic Cahn-Hilliard equation with improved conservation (ACH-IC) for approximating anisotropic surface diffusion. The ACH-IC model employs a new conserved quantity that approximates a step function more effectively, and yields second-order volume conservation while preserving energy dissipation for the classical anisotropic surface energy.  The second-order volume conservation as well as the convergence to the sharp-interface surface diffusion dynamics is derived through comprehensive asymptotic analysis. Numerical evidence not only reveals the underlying physics of the proposed model in comparison with the classical one, but also demonstrates its exceptional performance in simulating anisotropic surface diffusion dynamics.
\end{abstract}
%%%% Keyword entries to be placed here %%%%
% \keywords{Surface diffusion; volume conservation; asymptotic analysis; anisotropic Cahn-Hilliard equation； Onsager's variational principle}
%%%%%%%%%%%%%%%%%%%%%%%%%%%

\section{Introduction}
Surface diffusion is a fundamental process that plays an essential role in determining the behavior and properties of materials in various fields~\cite{sutton1995interfaces,doll1987recent,ehrlich1980surface,li1999numerical,xia1997finite,jiang2020}. It is a type of diffusion that occurs at the interface between two materials and is characterized by the motion of atoms or molecules along the surface \cite{ehrlich1980surface}.
Due to the difference in lattice orientations at material surfaces/interfaces in solids, the surface energies in solid materials usually present a property of anisotropy, which leads to anisotropic evolution processes. These are very common in solid materials \cite{xiao1991anisotropic,sutton1995interfaces}.

%The behavior of materials undergoing surface diffusion is of significant interest in many applications. For example, in the field of materials science, surface diffusion plays a critical role in determining the growth and properties of thin films, which are used in a wide range of technological applications\cite{sutton1995interfaces}. In chemistry, surface diffusion is involved in many important processes such as catalysis and surface reactions\cite{doll1987recent,ehrlich1980surface}. In physics, surface diffusion is important for the study of many phenomena such as crystal growth, wetting, and adhesion \cite{li1999numerical,xia1997finite}.

The mathematical formulation of (anisotropic) surface diffusion can be expressed as a geometric flow, as introduced by Mullins\cite{mullins1957theory}:
\begin{equation}\label{velocity}
	V_n = \Delta_s \mu,
\end{equation}
where $V_n$ represents the normal velocity of a closed $(d-1)$-dimensional surface $\Gamma(t)$. The operator $\Delta_s := \nabla_s \cdot \nabla_s$ denotes the surface Laplace-Beltrami operator, where $\nabla_s$ represents the surface gradient with respect to $\Gamma(t)$. The chemical potential $\mu$ is obtained through the thermodynamic variation of the total interfacial free energy given by\cite{Jiang17,jiang2020}
\begin{equation}\label{original_aniso_energy}
	W(\Gamma) = \int_\Gamma \gamma(\mathbf{n}) \mathrm{d}s,
\end{equation}
where $\gamma(\mathbf{n})$ represents the anisotropic surface energy density function with respect to the unit outward normal vector $\mathbf{n} = (n_1, n_2, \ldots, n_d)^T $. 
To express the thermodynamic variation of \eqref{original_aniso_energy}, we extend the anisotropic energy density $\gamma(\mathbf{n}):~ \mathbb{S}^{d-1} \rightarrow \mathbb{R}^{+}$ to a homogeneous function $\hat{\gamma}(\mathbf{p}): \mathbb{R}_*^d :=~ \mathbb{R}^d \backslash\{\mathbf{0}\} \rightarrow \mathbb{R}^{+}$, satisfying  (i) $\hat{\gamma}(c\mathbf{p}) = c\hat{\gamma}(\mathbf{p})$ for $c > 0$; (ii) $\hat{\gamma}(\mathbf{p})|_{\mathbf{p} = \mathbf{n}} = \gamma(\mathbf{n})$. A typical homogeneous extension \cite{deckelnick2005computation,jiang2019sharp} is
\begin{equation}\label{Cahn_transform}
	\hat{\gamma}(\mathbf{p}) := |\mathbf{p}| \gamma\left(\frac{\mathbf{p}}{|\mathbf{p}|}\right), \quad \forall \mathbf{p} \in \mathbb{R}_*^d.
\end{equation}
% This extended anisotropic density satisfies two properties: 
% (i) $\hat{\gamma}(c\mathbf{p}) = c\hat{\gamma}(\mathbf{p})$ for $c > 0$; (ii) $\hat{\gamma}(\mathbf{p})|_{\mathbf{p} = \mathbf{n}} = \gamma(\mathbf{n})$. \\
%We assume that the anisotropic density $\hat{\gamma}\in C^2(\mathbb{R}_*^d)$ is convex in this paper.
Next, we need to introduce the Cahn-Hoffman ${\bm\xi}$-vector\cite{hoffman1972vector,wheeler1999cahn} as:
\[
{\bm\xi}(\mathbf{n}) := \nabla\hat{\gamma}(\mathbf{p})\big|_{\mathbf{p}=\bn}.
\]
With these definitions, the chemical potential $\mu$ can be obtained by the thermodynamic variation as $\mu = \nabla_s \cdot {\bm\xi}$\cite{hoffman1972vector}, and the geometric flow in \eqref{velocity} is reformulated as:
\begin{equation}\label{geometric_flow}
	V_n = \Delta_s (\nabla_s \cdot {\bm\xi}).
\end{equation}
This geometric flow conserves the volume of the region enclosed by the surface $\Gamma(t)$ and decreases the total free energy $W(\Gamma)$ over time.
% This geometric flow exhibits two important properties \cite{bao2017stable,jiang2019sharp}:

% 1. The volume of the region enclosed by the surface $\Gamma(t)$ is conserved.

% 2. The free surface energy $W(\Gamma)$ decreases over time.

Although the geometric flow \eqref{velocity} or \eqref{geometric_flow} offers the accurate mathematical description of anisotropic surface diffusion, numerically solving it could be challenging and often requires complex numerical techniques to preserve the former two important properties as well as good mesh quality\cite{jiang21jcp}. Moreover, the pinch-off phenomenon often takes place during the dynamic evolution, i.e., a continuous film is split into two or more parts. In general, it is quite challenging for sharp-interface approaches to deal with topological change events, especially for three-dimensional cases. In order to tackle these difficulties, phase-field approaches as approximations to \eqref{geometric_flow} are proposed in the literature for simulating surface diffusion problems\cite{gugenberger2008comparison,jiang2012phase}. We consider a phase-field model with the free energy functional proposed by Torabi {\it {et al.}}\cite{torabi2009new}
\begin{equation}\label{Torabi_energy}
	E(u) = \int_{\Omega} \frac{\gamma(\boldsymbol{n})}{\varepsilon}\left(F(u)+\frac{\varepsilon^2}{2}|\nabla u|^2\right)\mathrm{d}x,\quad u\in\mathbb{R}^d,
\end{equation}
where  $u\in[-1,1]$ is the difference in the concentration fraction of the two components, %hence $-1\leqslant u\leqslant1 $, 
and its extreme values $\pm1$ correspond to the pure components (also called bulk phases). Moreover, $ F(u) = (u^2-1)^2/4 $ represents the Ginzburg-Landau double-well potential, $\varepsilon$ is the interfacial thickness parameter, and $\boldsymbol{n}$ is the normal vector defined by $\boldsymbol{n} = \nabla u/|\nabla u|$. The iso-contour $u=0$ is typically used to implicitly capture the interface's motion.
By applying traditional $H^{-1}$ gradient flow and rescaling by a factor $\varepsilon$, the classical anisotropic Cahn-Hilliard equation (ACH) is obtained:
\begin{align}\label{eq_M}
	\begin{cases}
		\begin{aligned}
			& u_t = \varepsilon^{-1}\nabla \cdot (M(u) \nabla \mu), \\
			&\mu =\frac{\delta E}{\delta u}= \varepsilon^{-1}\gamma(\boldsymbol{n}) F'(u) - \varepsilon \nabla \cdot \mathbf{m}(u),\\
			&\mathbf{m}(u) = \gamma(\boldsymbol{n})\nabla u + \frac{(I-\boldsymbol{n}\boldsymbol{n}^T)\nabla_{\boldsymbol{n}}\gamma(\boldsymbol{n})}{|\nabla u|}\left( \frac{1}{2}|\nabla u|^2+\frac{1}{\varepsilon^2}F(u)\right),
		\end{aligned}
	\end{cases}
\end{align}
where $M(u) = (1-u^2)^2$ is a degenerate mobility function that ensures its sharp-interface limit recovering pure surface diffusion \eqref{geometric_flow}\cite{lee2015degenerate,dziwnik2017anisotropic}.
It is noted that the term $F(u)/|\nabla u|$ in \eqref{eq_M} can be problematic since $|\nabla u| $ could vanish in some regions where $F(u)$ does not. There are mainly two ways to handle this problem. One method is using the asymptotic result $F(u)\sim\frac{\varepsilon^2}{2}|\nabla u|^2$  nearby the interface  to yield an approximation\cite{wise2007solving,torabi2009new}, and the other approach is regularizing the normal vector $\mathbf{n}$ as $\bhatn:=\nabla u/(\sqrt{|\nabla u|^2+\varepsilon^2})$\cite{salvalaglio2021doubly,salvalaglio2021doubly1}. In this paper, we take the latter strategy.

Classical phase-field models, while widely used for simulating interface problems such as surface diffusion and two-phase flow, suffer from several critical limitations. One  issue is their low resolution in terms of $\varepsilon$ in capturing interface properties, necessitating the adoption of an extremely small $\varepsilon$ to accurately track the interface's position and temporal evolution. This not only increases computational costs but also poses challenges for efficiently simulating long-term dynamics or large-scale systems\cite{salvalaglio2017morphological,salvalaglio2017phase,geslin2019phase}. Additionally, these models often fail to preserve the enclosed area (or volume) of the interface during evolution, leading to physical inconsistencies.  Especially, these models exhibit a spontaneous shrinkage effect\cite{2017Spontaneous}: when the radius of the simulated interface becomes smaller than a critical value, the interface would gradually disappear. This phenomenon is highly non-physical and limits the applicability of such models for problems involving small or thin interfaces, where accurate preservation of the enclosed volume is essential. 

Recently, there are a lot of works (e.g., Refs.~\cite{ratz2006surface,salvalaglio2021doubly,salvalaglio2021doubly1,bretin2022approximation}) dedicated to improving the approximation capabilities of the phase-field model. These can be summarized into three approaches.
%Specifically, there are mainly two ways to modify \eqref{eq_M} to achieve a better approximation. 
One approach is by introducing an additional degenerate coefficient $g$ into the chemical potential $\mu$\cite{ratz2006surface}, i.e., modifying the second equality of \eqref{eq_M} as
\[
g(u)\mu = \gamma(\bhatn) F'(u) - \varepsilon^{2} \nabla \cdot \mathbf{m}(u),
\]
where $g(u)=(1-u^2)^p$, $p\geqslant1$. Although this approach improves the approximation effect of the phase-field model, it violates the variational structure which makes the model thermodynamically inconsistent. 
%makes it harder to prove properties of solutions and may incur certain numerical stability issues. 
The other approach is to impose a singularity in the formulation of a new energy functional (also called de Gennes-Cahn-Hilliard (dGCH) energy)\cite{salvalaglio2021doubly1}:
\begin{equation}\label{dGCH}
	\hat{E}(u) = \int_{\Omega} \frac{1}{g(u)}\frac{\gamma(\bhatn)}{\varepsilon}\left(F(u)+\frac{\varepsilon^2}{2}|\nabla u|^2\right)\mathrm{d}x,\quad u\in\mathbb{R}^d.
\end{equation}
%The dGCH energy is $\Gamma$-converging to the interfacial energy \eqref{original_aniso_energy} when $\gamma(\mathbf{n})\equiv 1 $, which indicates its minimizer is converging to the minimizer of the sharp-interface model as $\varepsilon\to0$. 
Formal asymptotic analysis in the isotropic case\cite{salvalaglio2021doubly} and numerical simulations\cite{salvalaglio2021doubly1} show that the equation obtained by the $H^{-1}$ gradient flow of the dGCH energy \eqref{dGCH} can approximate surface diffusion more accurately than some other phase-field approximations. However, it complicates the model by introducing more nonlinear terms, especially in anisotropic cases, leading to difficulties in its theoretical and numerical analysis.
% does not preserve the original total free energy \eqref{Torabi_energy} dissipation. 
%Moreover, to our knowledge, there exists no reasonable explanation in the literature why this model can approximate anisotropic surface diffusion more accurately in theory.
Instead of modifying the model from the energetic perspective, Bretin et al.\cite{bretin2022approximation} proposed a third method by considering the gradient flow of \eqref{Torabi_energy} under a special metric:
\[
<f_1,f_2>_{H_0^1}=\int_\Omega M(u)\nabla\left(\frac{1}{g(u)} f_1\right)\cdot\nabla\left(\frac{1}{g(u)} f_2\right).
\]
This amounts to a different dissipation structure while preserving the same energy as in \eqref{Torabi_energy}. 
Previous phase-field approaches are mainly based on the conservation of the phase variable $u$, which does not give a good approximation to the volume enclosed by the interface. As a result, such phase-field models may give rise to undesirable effects---most notably, spontaneous shrinkage of an interface\cite{2017Spontaneous}: especially when initialized in a steady-state configuration, the interface can gradually diminishes over time. To overcome this difficulty, we propose a new type of phase-field model with improved volume conservation by introducing a new conserved quantity $Q(u)$ that enhances the approximation to the enclosed volume. It can be shown that the newly proposed model for anisotropic Cahn-Hilliard equation with improved conservation (ACH-IC) can achieve ``second-order accurate real volume conservation'', i.e., preserving the enclosed area (or volume, in higher dimensions) of the region $\{x : u(x)>0\}$ at a second order in $\varepsilon$. This theoretical enhancement is also numerically validated in both isotropic and anisotropic cases, demonstrating the superior performance of the ACH-IC model over earlier methods.

It is worth mentioning that in the case of isotropic energy functional and a specific quadratic mobility, the ACH-IC model recovers the phase-field model proposed by Bretin\cite{bretin2022approximation}. However, 
%beyond the second-order accuracy in Cahn-Hilliard model, 
the framework in this paper provides a clear intuition and deep physical understanding on the model improvement. The technique of introducing an improved conserved quantity is not limited to the anisotropic Cahn-Hilliard equation. In fact, by using a properly modified conservation equation, our work offers a novel framework for the phase-field simulation of interface problems with improved volume conservation, leading to much better approximations to the corresponding sharp-interface dynamics. This framework is clearly demonstrated from the OVP perspective: with different energy functionals and different mobilities or dissipation functionals utilized, variational derivation can be developed in the constraint of conservation equations to obtain phase-field models for a wide variety of problems. For instance, one can introduce elastic or interfacial energy terms into the energy functional to model phenomena such as solid-state dewetting\cite{jiang2012phase,Boccardo2022}. Additionally, by incorporating terms dependent on $\nabla u$ into the mobility $M(u)$, it is straightforward to model anisotropic surface diffusion, with a directionally dependent mobility\cite{taylor1994linking,garcke2023diffuse}.

The rest of this paper are structured as follows. In Section \ref{chapter2}, we derive the proposed model via the OVP and the gradient flow. In Section \ref{chapter3}, we employ matched asymptotic expansions to derive the second-order approximation solution, recover the pure surface diffusion, and establish the second-order real volume conservation. In Section \ref{chapter4}, we present a series of numerical experiments to validate the better approximation properties of our model. Finally, concluding remarks are made in Section \ref{chapter5}.

\section{Variational derivation of the ACH-IC model}\label{chapter2}
\subsection{Onsager's variational principle}
The Onsager variational principle (OVP), first formulated by Onsager\cite{onsager1931reciprocal1,onsager1931reciprocal2}, provides a powerful framework for describing non-equilibrium processes within the linear response regime of irreversible thermodynamics. This principle has been widely utilized in diverse fields, including fluid dynamics\cite{qian2006variational,doi2019application,zhang2022variational}, soft matter physics\cite{doi2011onsager,wang2021onsager}, and interfacial problems such as solid-state dewetting\cite{jiang2019application,zhao2024dynamics}. We first present a brief review of this principle.

Consider an isothermal system not far from equilibrium, with state variables  
\[  
\beta(t) = (\beta_1(t), \beta_2(t), \ldots, \beta_n(t)),  
\]  
and their rates of change  
\[  
\dot{\beta}(t) = (\dot{\beta}_1(t), \dot{\beta}_2(t), \ldots, \dot{\beta}_n(t)),  
\]  
where $\dot{\beta}$ represents the time derivative of $\beta$.
The system's total free energy is denoted by \(W(\beta)\). OVP states that the dynamics of this system can be derived by minimizing the Rayleighian defined by
\[  
\mathcal{R}(\beta, \dot{\beta}) = \dot{W}(\beta, \dot{\beta}) + \Phi(\dot{\beta}, \dot{\beta}),  
\]  
where \(\dot{W}(\beta, \dot{\beta}) = \sum_{i=1}^n \frac{\partial W}{\partial \beta_i} \dot{\beta}_i\) represents the rate of change of the free energy, and \(\Phi(\dot{\beta}, \dot{\beta})\) is the dissipation function given by  
\[  
\Phi(\dot{\beta}, \dot{\beta}) = \frac{1}{2} \sum_{i=1}^n \sum_{j=1}^n \lambda_{ij}(\beta) \dot{\beta}_i \dot{\beta}_j,
\]  
in the linear response regime. Physically, $\Phi(\dot{\beta}, \dot{\beta})$ is half the rate of free energy dissipation.
Here, the friction coefficients \(\{\lambda_{ij}\}\) form a symmetric, positive definite matrix. Minimizing \(\mathcal{R}\) with respect to \(\dot{\beta}_i\) yields the dynamic equations  
\[  
-\frac{\partial W}{\partial \beta_i} = \sum_{j=1}^n \lambda_{ij} \dot{\beta}_j, \quad i = 1, 2, \ldots, n,  
\]  
which state the balance between reversible and dissipative forces. This balance ensures thermodynamic consistency and establishes a direct link between the system's free energy landscape and its non-equilibrium dynamics.
In this work, we shall utilize OVP to derive a phase-field model for anisotropic surface diffusion.

\subsection{A new conserved quantity and the model derivation}
In this subsection, we derive a variational model by considering a more general conserved quantity and employing the OVP.
In previous phase-field models\cite{gugenberger2008comparison,salvalaglio2015faceting,torabi2009new,backofen2019convexity} for anisotropic surface diffusion, the conservation constraint is typically expressed as the conservation of the phase-field variable $u$, i.e., 
\begin{equation}\label{integral_conservation}
	\int_\Omega u\dx=\text{constant},
\end{equation}
which approximately yields the conservation of the volume of the region enclosed by the zero-level set of the phase function, i.e.,
\[
\Omega^+(t) := \{ x : u(x,t) > 0 \}.
\]
%This approximation is only accurate when $u$ approximates the indicator function of $\Omega^+(t)$ well. 
It has been shown that\cite{2017Spontaneous,bretin2022approximation} in general the volume conservation property of phase-field models under the constraint \eqref{integral_conservation} is of first-order accuracy:
\[
\Omega^+(t) = \Omega^+(0) + \mathcal{O}(\varepsilon),
\]
which is a result of the approximation property of the conserved quantity \( u \) to a step function \( S(x) = \chi_{_{\Omega^+}}(x) - \chi_{_{\Omega \backslash \Omega^+}}(x) \) with a \textit{tanh} profile as
\[
u(x, t) = \tanh\left(\frac{\text{dist}(x, \partial\Omega^+(t))}{\sqrt{2}\varepsilon}\right) + \mathcal{O}(\varepsilon).
\]

To improve the quality of volume conservation, we introduce a more general conserved quantity \( Q(u) \) to replace \( u \). This is achieved by considering the following constraint
\begin{equation}\label{manifold_M}
	\int_\Omega Q(u) \, \mathrm{d}x = \text{constant}.
\end{equation}
%This approach provides greater flexibility in capturing the physical behavior of the system, ensuring a more accurate approximation. 
We require that the conserved variable \( Q(u) \) approximates the step function \( S(x) \) more closely than \( u \) when \( u \) is a \textit{tanh} profile. There is flexibility of selecting $Q(u)$, for instance, $Q(u)$ can be defined by 
\[
Q(u) = \frac{\int_{0}^{u} (1 - z^2)^k \, \mathrm{d}z}{\int_{0}^1 (1 - z^2)^k \, \mathrm{d}z},
\]
where \( k \geqslant 1 \) is an integer. This quantity converges to the step function \( S(x) \) more rapidly as \( \varepsilon \) decreases when \( u = \tanh\left(\frac{\text{dist}(x, \partial \Omega^+)}{\sqrt{2} \varepsilon}\right) \). In fact, it can be observed that as \( \varepsilon \to 0 \), for \( x \notin \partial \Omega^+ \), \( u(x) \to \pm 1 \) and \( Q(u(x)) \to \pm 1 \). Applying L'Hopital's Rule, we obtain:
\[
\lim_{\varepsilon \to 0} \frac{Q(u) - S}{u - S} = \lim_{\varepsilon \to 0} \frac{Q'(u) \frac{\partial u}{\partial \varepsilon}}{\frac{\partial u}{\partial \varepsilon}} = \lim_{\varepsilon \to 0} \frac{(1 - u^2)^k}{\int_0^1 (1 - z^2)^k \, \mathrm{d}z} = 0, \quad \text{for any } x \notin \partial \Omega^+.
\]

To obtain an intuitive understanding of the approximation behavior of \( Q \), we present a numerical comparison of \( Q \) for different values of \( k = 1, 2, 3 \) in Fig.~\ref{Q_fig}. This figure visually demonstrates how the approximation to the step function is improved by increasing \( k \), showing the effectiveness of \( Q \) in capturing the sharp transition between phases. 
%Readers can observe the differences in the sharpness of the transition, which highlights the improved performance of the proposed model with higher values of \( k \).
\begin{figure}[htb]
	\centering
	\includegraphics[trim=0.2cm 0.0cm 0.2cm 0.2cm, clip,width=0.7\linewidth]{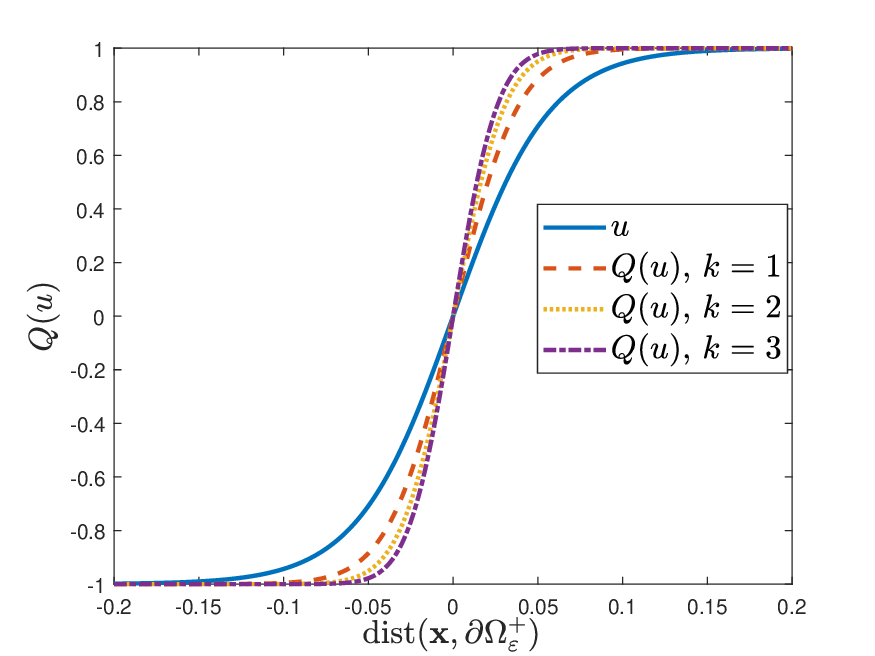}
	\caption{$Q(u)$ for different $k$, when $u=\tanh\left(\frac{\text{dist}(x,\partial\Omega^+)}{\sqrt{2}\varepsilon}\right)$ with $\varepsilon=0.05$.}
	\label{Q_fig}
\end{figure}

Given the new conserved quantity \( Q \), the conservation law can be expressed as
\begin{equation}\label{conservation_law}
	\partial_t Q+\nabla \cdot \mathbf{J}=0, \quad \Longleftrightarrow \quad u_t=-\frac{1}{Q'(u)} \nabla \cdot \mathbf{J},
\end{equation}
where \( \mathbf{J} \) represents the concentration flux.
Before we can apply the OVP under the conservation law \eqref{conservation_law}, we need to express the free energy and the dissipation functionals. In this work, we always use the Torabi-type free energy functional \eqref{Torabi_energy}, and then the rate of change of the total free energy $E_t$ can be calculated as
\begin{equation}
	\label{changeE}
	E_t(u)=\int_\Omega \frac{\delta E(u)}{\delta u}u_t\dx
	%=-\int_\Omega \frac{\delta E(u)}{\delta u}(N(u)\nabla\cdot \mathbf{J})\dx
	=\int_\Omega \nabla\left(\frac{1}{Q'(u)}\frac{\delta E(u)}{\delta u}\right)\cdot \mathbf{J}\dx,
\end{equation}
where we have assumed the no-flux boundary condition $\mathbf{J}\cdot\mathbf{\hat{n}}=0$ on $\partial\Omega$. 
The dissipation functional $\Phi$ is proportional to the square of
the concentration flux $\mathbf{J}$ and can be written as
\begin{equation}
	\label{dissipation}
	\Phi:=\int_\Omega \frac{\varepsilon}{2M(u)}|\mathbf{J}|^2\dx,
\end{equation} 
The function \( 1/M(u) \) plays a role of friction coefficient, measuring the hardness for atoms to move in the bulk phases (where \( u = \pm1 \)). For example, when we take $M(u)=(1-u^2)^l$ with $l \geqslant 1$, the dissipation functional $\Phi$ is consistent with the ones in  Refs.~\cite{qian2006variational,jiang2019application}. 
High-order degeneracy in \( M(u) \) (i.e., \( l \geqslant 2 \)) was typically required to achieve better performance in the simulation\cite{lee2015degenerate}. However, as we will see in our discussions and numerical results, the condition
\begin{equation}
\label{condition_mobility}
    1\leqslant k\leqslant l\leqslant2k+1,
\end{equation}
is sufficient to achieve accurate approximations in our proposed model, which allows more flexible choices in selecting $M(u)$.

Combining the time derivative of the free energy \eqref{changeE} and the dissipation functional \eqref{dissipation}, the Rayleighian is given by:
\begin{equation}
	\label{Raylei}
	\mathcal{R}:=E_t+\Phi=\int_\Omega \nabla\left(\frac{1}{Q'(u)}\frac{\delta E(u)}{\delta u}\right)\cdot \mathbf{J}\dx+\int_\Omega \frac{\varepsilon}{2M(u)}|\mathbf{J}|^2\dx.
\end{equation}
To find the constitutive equation for the concentration flux \( \mathbf{J} \), we minimize the Rayleighian with respect to \( \mathbf{J} \). The minimum is achieved when
\[
\mathbf{J}=-\frac{1}{\varepsilon}M(u)\nabla\left(\frac{1}{Q'(u)}\frac{\delta E(u)}{\delta u}\right).
\]
Combining this expression with the conservation law \eqref{conservation_law}, we derive the evolution equation for \( u \):
\begin{equation}\label{eq_NMN}
	\begin{cases}
		\begin{aligned}
			& u_t= \varepsilon^{-1}N(u)\nabla \cdot(M(u) \nabla( N(u)\mu)), \\
			&\mu=\frac{\delta E}{\delta u}=\varepsilon^{-1}\gamma(\bhatn) F'(u)-\epsilon \nabla \cdot \mathbf{m}(u),\\
			&\mathbf{m}(u)=\gamma(\bhatn)\nabla u+\frac{(I-\bhatn\bhatn^T)\nabla_{\bhatn}\gamma(\bhatn)}{\sqrt{|\nabla u|^2+\varepsilon^2}}\left( \frac{1}{2}|\nabla u|^2+\frac{1}{\varepsilon^2}F(u)\right),
		\end{aligned}
	\end{cases}
\end{equation}
where $N(u)=1/Q'(u)$, $M(u)=(1-u^2)^l$, and the regularization technique has been used in approximating $|\nabla u|$ to avoid singularity in $\mathbf{m}(u)$.
In this paper, we always take $Q'(u)=\frac{ (1 - u^2)^k }{\int_0^1 (1 - z^2)^k\mathrm{d}z}$. We would refer to the anisotropic Cahn-Hilliard equation with improved conservation \eqref{eq_NMN} as \textbf{ACH-IC}. It is also noteworthy that many existing models can be obtained as special cases of ACH-IC. When $k=1$ and $l=2$ and isotropic surface tension is used (i.e. $\gamma$ is independent of $\bhatn$), the phase-field model in Ref.~\cite{bretin2022approximation} is recovered from the ACH-IC model. Alternatively, if we take $k=l=0$, we obtain the classical Cahn-Hilliard equation with constant mobility.
% In general, $k$ and $l$ can be any nonnegative integers. 
% We will assume $1\leqslant k\leqslant l\leqslant2k+1$ throughout this paper (this condition would be used in the Section \ref{chapter3}).

Physically, it is important to highlight the distinct roles played by \( M(u) \) and \( N(u) \) in the proposed model. \( M(u) \) controls the diffusion rate and confines diffusion mainly to the interfacial region. On the other hand, \( N(u) \) improves the solution's proximity to a step function, ensuring better real volume conservation. Intuitively, although the first equation in \eqref{eq_NMN} is no longer a conservation equation over the whole domain, the singularity $u=\pm1$ in $N(u)=1/Q'(u)$ forces zero flux from the interface region into the bulk. This will be made clear by the asymptotic analysis and the numerical demonstration in later sections. The combination of \( M(u) \) and \( N(u) \) serve a more accurate description of interface evolution through robust modeling of anisotropic surface diffusion.

%After introducing the proposed model, we briefly address its mathematical foundation. By following the framework established in Ref.~\cite{niu2024weak}, the existence of weak solutions can be similarly demonstrated, ensuring the well-posedness of the model and providing a robust basis for further theoretical and numerical analysis.

A potential issue, however, arises when \( u \) approaches the singularities of \( N(u) \) (i.e., \( u \in \{-1, 1\} \)). Although the ideal solution is expected to maintain a smooth \textit{tanh} profile, the model lacks a strict maximum principle, which allows \( u \) to attain values of \(-1\) or \( 1 \) in practice. To address this issue, we regularize $N(u)$ as  
\[  
\tilde{N}(u) = \frac{\int_0^1(1-z^2)^k\mathrm{d}z}{|1-u^2|^k+\varepsilon^{2k}},  
\]  
which effectively avoids singularities. 
Importantly, this adjustment does not affect the asymptotic behavior of the model. As demonstrated in the analysis part in Section \ref{chapter3} (also see Ref.~\cite{bretin2022approximation}), the additional terms in the expansion arise only at third or higher orders and have no impact on the leading-order normal velocity. For simplicity, we therefore retain the original definitions  \( N(u) = 1/Q'(u) \) in our analysis, under the assumption that \( u \neq -1, 1 \) in the expansions.

%\subsection{Gradient flow approach}
At the end of this subsection, we provide a mathematical derivation of the ACH-IC. Starting from the standard \( H^{-1} \) gradient flow with respect to the conserved quantity \( Q \), we can express the conservation law as:
\[
Q_t = \nabla \cdot \left(M(u)\nabla \mu_Q\right),
\]
where \( M(u) \) is the mobility and \( \mu_Q \) represents the chemical potential defined by
\[
\mu_Q = \frac{\delta E}{\delta Q} = \frac{\delta E}{\delta u} \cdot \frac{\mathrm{d}u}{\mathrm{d}Q} = \frac{1}{Q'(u)} \frac{\delta E}{\delta u}=N(u)\mu.
\]
Since $Q_t=Q'(u)u_t$, we arrive at the ACH-IC. 
This shows that by considering a more general conserved quantity \( Q(u) \), new model with improved conservation can be naturally derived from the standard gradient flow formulation.

%Moreover, we can also derive the Q-ACH model using the Riemannian gradient flow approach, as discussed in Otto's work on the geometry of dissipative systems \cite{otto2001geometry}. The derivation follows a similar procedure to the one presented by Bretin et al. \cite{bretin2022approximation}, and the gradient flow is calculated with respect to the variable \( u \). For more detailed steps and mathematical justifications, we refer the reader to the appendix, where the full derivation is outlined.

\subsection{Properties of the ACH-IC}
As a result of the variational structure,
%brings many valuable properties.
the ACH-IC model enjoys the energy decaying property: assuming periodic or natural boundary conditions on \( \partial\Omega \), with respect to the classic anisotropic energy \eqref{Torabi_energy}, we have 
\[
\begin{aligned}
	\frac{\mathrm{d}}{\dt} E(u)  &=\int_\Omega \frac{\delta E}{\delta u}u_t \dx=\varepsilon^{-1}\int_\Omega \mu N(u)\nabla\cdot\left(M(u)\nabla\left(N(u)\mu\right)\right)\dx\\
	&=-\varepsilon^{-1}\int_\Omega  M(u)\left|\nabla(N(u)\mu)\right|^2 \dx\leqslant 0.
\end{aligned} 
\]
Moreover, the introduction of the new conserved quantity still preserves the asymptotic properties:
\begin{enumerate}
	\item  Let  $\Omega^+_{\varepsilon}(t) = \left\{ u_{\varepsilon}(\cdot, t) > 0 \right\}$. The solution $u_{\varepsilon}$ to \eqref{eq_NMN} has the following second-order uniform asymptotic approximation:
	\begin{equation}\label{second_u}
		u_{\varepsilon} = u_{unif,0}+\varepsilon u_{unif,1}+\mathcal{O}(\varepsilon^2),
	\end{equation}
	where the leading-order approximation $ u_{unif,0}$ is a \textit{tanh} profile given by
	\[
	u_{unif,0}= \tanh\left( \frac{\text{dist}\left(x, \partial \Omega_\varepsilon^+(t)\right)}{\sqrt{2}\varepsilon} \right),
	\]
	and there exists a constant $C$ such that the first-order approximation is uniformly bounded by $C$, i.e., $ |u_{unif,1}|<C$. Moreover, $ u_{unif,1}\to0$ as $x$ is away from the interface $\partial\Omega_\varepsilon^+$.
	
	\item   The normal velocity of the interface follows the law of anisotropic surface diffusion:
	\begin{equation}\label{pure_surfacediffusion}
		V_{\varepsilon} = C_l\Delta_{s} \left( \nabla_s \cdot \bxi(\nu) \right) + \mathcal{O}(\varepsilon),
	\end{equation}
	where $C_l$ is a constant that is related to the mobility $M(u)=(1-u^2)^l$, defined by
	\[
	C_l=\frac{2}{3}\int_{0}^{1}(1-z^2)^{l-1}\mathrm{d}z.
	\]
	
	\item When $\partial \Omega^+_{\varepsilon}(t)$ is finite and possesses $\mathcal{C}^{2}$ regularity at all time, the real volume is conserved up to second order:
	\begin{equation}\label{second_volume}
		\left| \Omega^+_{\varepsilon}(t) \right| = \left| \Omega^+_{\varepsilon}(0) \right| + \mathcal{O}\left( \varepsilon^{2} \right).
	\end{equation}
	
\end{enumerate}
These properties collectively provide the theoretical guarantee for the improvement of the ACH-IC model in its approximation ability to the sharp-interface anisotropic surface diffusion flow, with second-order accuracy in real volume conservation. Compared to existing phase-field approaches, the ACH-IC model mitigates the significant limitation of spontaneous shrinkage by achieving more accurate interface approximation and better physical consistency, particularly in the conservation of the enclosed domain's volume. In the following sections, we will show these properties by means of matched asymptotic analysis and illustrate their practical advantages through numerical experiments.

Before we close this subsection, we give some discussions about the strongly anisotropic cases (i.e., $\hat{\gamma}$ is non-convex). This problem would be ill-posed due to the
anti-diffusion in $\nabla\cdot\mathbf{m}$\cite{wise2007solving}. In order to tackle the ill-posedness of the dynamic problem arising from the non-convex energy density, a commonly used approach is to introduce a high-order Willmore regularization term 
$\frac{\beta}{2}\int_\Omega\frac{1}{\varepsilon^3}(F'(u)-\varepsilon^2\Delta u)^2\dx$\cite{chen2013efficient,makki2016existence} into the energy functional \eqref{Torabi_energy}, where $\beta$ controls the regularization strength. This results in a modified chemical potential given by
\[
\mu=\frac{1}{\varepsilon}\gamma(\bhatn) F'(u)-\epsilon \nabla \cdot \mathbf{m}(u)+\beta(\frac{1}{\varepsilon^2}F''(u)\omega-\Delta\omega), \quad\omega=\frac{1}{\varepsilon}F'(u)-\varepsilon\Delta u.
\]
In the strongly anisotropic case, the ACH-IC retains all the previously mentioned properties, except for the normal velocity. This is because the Willmore regularization terms dominate the leading-order solution, resulting in a modified, regularized normal velocity. The analysis of these properties follows similar arguments as those in the unregularized case (also see Ref.~\cite{torabi2009new}) and will be omitted in this work for brevity.
%However, it is worth mentioning that NMN-ACH model can still give a more accurate approximation for surface diffusion numerically, and we give their comparison in Section 6.  

\section{Analysis of the ACH-IC model}\label{chapter3}
We first employ the matched asymptotic analysis to obtain the uniform asymptotic approximation \eqref{second_u} and recover the sharp-interface limit of the ACH-IC \eqref{pure_surfacediffusion}. 
Our asymptotic analysis will follow similar techniques in the literature for degenerate Cahn-Hilliard equations\cite{cahn1996cahn,bretin2022approximation}, and for the anisotropic case\cite{dziwnik2017anisotropic,garcke2023diffuse}.
Since the normal velocity \eqref{geometric_flow} is represented by the Cahn-Hoffman vector ${\bm\xi}  $, we need to establish some useful equalities. For any vector $ \mathbf{p}\in\mathbb{R}_*^d $ and $ \lambda>0 $, we have $ \hat{\gamma}(\lambda\mathbf{p})=\lambda\hat{\gamma}(\mathbf{p}) $. Differentiating this equality with respect to $ \mathbf{p} $ and $ \lambda $, we can get the following equalities:
\begin{equation}\label{equ_xi}
	{\bm\xi}(\lambda\mathbf{p})={\bm\xi}(\mathbf{p}),\quad\nabla{\bm\xi}(\lambda\mathbf{p})=\frac{1}{\lambda}\nabla{\bm\xi}(\mathbf{p}),\quad{\bm\xi}(\mathbf{p})\cdot\mathbf{p}=\hat{\gamma}(\mathbf{p}).
\end{equation}
Straightforward calculation gives the representation in terms of the Cahn-Hoffman vector:
\[
\begin{aligned}
	(I-\bhatn\bhatn^T)\nabla_{\bhatn}\gamma(\bhatn)=(I-\bhatn\bhatn^T){\bm\xi}(\bhatn)={\bm\xi}(\bhatn)-\bhatn\hat{\gamma}(\bhatn).
\end{aligned}
\]
Then we can reformulate the ACH-IC model \eqref{eq_NMN} into
\begin{numcases}{}
	\partial_t u=\frac{1}{\varepsilon} N \nabla \cdot(M \nabla (N\mu)),\label{eq1} \\
	\mu=\frac{1}{\varepsilon}\hat{\gamma}\left(\bhatn\right) F'(u)-\epsilon \nabla \cdot \mathbf{m}(u),\label{eq2}\\
	\mathbf{m}(u)=\hat{\gamma}\left(\bhatn\right)\nabla u+\frac{\left({\bm\xi}(\bhatn)-\bhatn\hat{\gamma}(\bhatn)\right)}{\sqrt{|\nabla u|^2+\varepsilon^2}}\left( \frac{1}{2}|\nabla u|^2+\frac{1}{\varepsilon^2}F(u)\right).\label{eq3}
\end{numcases}

As $\varepsilon\to0$, there is a transition layer near the zero-level set $\Gamma_\varepsilon(t):=\{\mathbf{x}\in\Omega:u_\varepsilon(\mathbf{x},t)=0\}$, where the variable $u$ undergoes significant changes. In this section, we assume that $\Gamma_\varepsilon$ is a smooth closed surface and divide the domain $\Omega$ into two regions: $\Omega^+_\varepsilon(t):=\{\mathbf{x}\in\Omega:u_\varepsilon(\mathbf{x},t)>0\}$ and $\Omega^-_\varepsilon(t):=\{\mathbf{x}\in\Omega:u_\varepsilon(\mathbf{x},t)<0\}$.
%It is worth mentioning that in strongly anisotropic cases (i.e., $\hat{\gamma}$ is not convex), the zero-level set surface $\Gamma_\varepsilon$ may not be smooth. To address this, we can introduce regularization techniques, such as the addition of a Willmore regularized term or bi-harmonic regularized term \cite{wise2007solving,chen2013efficient}, to smooth out $\Gamma$. By employing this technique, we can obtain similar results.

\subsection{Outer expansion}

% We assume that  $u=\pm1+\mathcal{O}(\varepsilon^2)$ in the outer region and verify a posteriori that this assumption is compatible for the matching with the inner solution in \eqref{eq1}-\eqref{eq3}. 
In the outer region, we can expand the following variables: 
\[
\begin{aligned}
	&u=u_0+\varepsilon u_1+\varepsilon^2 u_2\cdots,\quad
	&\mu=\varepsilon ^{-1}\mu_{-1}+\mu_0+\varepsilon\mu_1+\cdots,\quad
	&\bhatn=\bhatn_0+\varepsilon\bhatn_1+\cdots.
\end{aligned}
\]
As a consequence, 
\[
F'(u)=F'(u_0)+\varepsilon F''(u_0)u_1+\mathcal{O}(\varepsilon^2),\quad
	\hat{\gamma}(\bhatn)=\hat{\gamma}(\bhatn_0)+\varepsilon\bxi(\bhatn_0)\bhatn_1+\mathcal{O}(\varepsilon^2).
\]

To further expand the elliptic operator $N(u)\nabla\cdot\left(M(u)\nabla N(u)\right)$, we need to determine the order of $M(u)$ and $N(u)$. We first assume that $u_0=\pm1$ at which the energy is minimized in the bulks. This will be validated a posteriori to be a consistent assumption. Let $w$ be the minimum integer such that $u_w\neq0$. Straightforward calculation shows that the leading order terms of $M(u)$ and $N(u)$ are 
\[
\varepsilon^{lw}\frac{M^{(l)}(\pm1)}{l!}u_w^{lw}, \quad \varepsilon^{-kw}\frac{k!}{Q^{(k+1)}(\pm1)}u_w^{-kw}.
\]
Hence,
\[
N(u)\nabla\cdot\left(M(u)\nabla N(u)\right)=\varepsilon^{w(l-2k)}\mathcal{A}+o(\varepsilon^{w(l-2k)}),
\]
where $\mathcal{A}$ is an elliptic operator defined by
\[
\mathcal{A}=\frac{(k!)^2M^{(l)}(\pm1)}{l!(Q^{(k+1)}(\pm1))^2}\frac{1}{u_w^{kw}}\nabla\cdot\left(u_w^{lw}\nabla(\frac{1}{u_w^{kw}})\right).
\]

% We adopt the following notations to simplify the writing:
% \[
% M(u)=m_0+\varepsilon m_1+\varepsilon^2 m_2+\cdots,\quad N(u)=n_0+\varepsilon n_1+\varepsilon^2 n_2+\cdots.
% \]

Substituting these expansions into \eqref{eq1}--\eqref{eq3}, we can obtain the outer equation at each order. In particular, under the condition \eqref{condition_mobility} (or $l-2k<1+1/w$), we have
\begin{equation}
	\label{out_leading}
0=\mathcal{A}\mu_{-1},\quad
		\mu_{-1}=\hat{\gamma}(\bhatn_0)F'(u_0).
        %-\nabla\cdot\left(\frac{{\bm\xi}(\bhatn_0)-\bhatn_0\hat{\gamma}(\bhatn_0)}{\sqrt{|\nabla u|^2+\varepsilon^2}}F(u_0)\right).
\end{equation}
By the ellipticity of the operator $\mathcal{A}$, we find that $(u_0, \mu_{-1})=(\pm1,0)$ composes a solution to the leading order system \eqref{out_leading}. It should be noted that the common cases $(k,l)=(1,1)$ and $(k,l)=(1,2)$ satisfy the assumption, and thus our discussion also applies in such special cases.
%In accordance with the ansatz that our solution is nearly equilibrated in the outer region, we assume that $u_0=\pm1$ which is one solution in \eqref{out_leading}, and thus $\mu_{-1}=0$.

Continuing our discussion under the same assumption ($l-2k<1+1/w$), at the first order  we have
\begin{equation}
		\label{out_second}
	0=\mathcal{A}\mu_0,\quad
	\mu_0=\hat{\gamma}(\bhatn_0)F''(u_0)u_1.
\end{equation}
The ellipticity of the operator $\mathcal{A}$ gives that $\mu_{0}$ is constant (and so is $u_1$). Indeed, $(u_1, \mu_{0})=(0,0)$ as we will see after matching with the inner solution.

% We choose the solution  $(u_1, \mu_{0})=(0,0)$ that can only be matched with the inner solution.

\subsection{Inner expansion}
Let $d(\mathbf{x},t) = \mathrm{dist}(\mathbf{x}, \Gamma_\varepsilon(t))$ be a signed distance function between $\mathbf{x}$ and the surface $\Gamma_\varepsilon(t)$. When $d(\mathbf{x},t)>0$, $\mathbf{x}$ belongs to $\Omega^+_\varepsilon$. We introduce a rescaled inner variable $ \rho(\mathbf{x},t):=d(\mathbf{x},t)/\varepsilon $ along the normal direction to $\Gamma_\varepsilon(t)$, and take a local normal-tangent coordinate system with respect to the interface $\Gamma_\varepsilon(t)$, i.e.,
\[
\mathbf{x}=\mathbf{r}(\mathbf{s}(\mathbf{x},t),t)+\varepsilon\rho(\mathbf{x},t)\mathbf{\nu}(\mathbf{s}(\mathbf{x},t),t),
\]
where $ \mathbf{r}(\mathbf{s},t) $ is a parametrization of $\Gamma_\varepsilon(t)$, $ \mathbf{s}=(s_1,s_2,\cdots,s_{d-1}) $ is arc length, and $ \mathbf{\nu}(\mathbf{s},t) $ is the outward unit normal vector to $\Gamma_\varepsilon(t)$ at $ \mathbf{r}(\mathbf{s},t) $ pointing into $\Omega^+_\varepsilon$. Let 
\[
\mathbf{t}_i=\frac{\partial \mathbf{r}}{\partial s_i},\quad\quad \text{for }i=1,2,\cdots,d-1.
\]
Then $ \{\mathbf{t}_1,\mathbf{t}_2,\cdots,\mathbf{t}_{d-1}\} $ forms an orthonormal basis of the tangent space of $\Gamma_\varepsilon(t)$ at $ \mathbf{r}(\mathbf{s},t) $, and we have
\[
\nabla s_i=-\kappa_i\mathbf{\nu},\quad\quad \text{for }i=1,2,\cdots,d-1,
\]
where $\kappa_i$ are principal curvatures. As in Ref.~\cite{dai2014coarsening}, we can derive the identities
\[
\nabla d=\varepsilon \nabla \rho=\mathbf{\nu}(\mathbf{s}, t),\quad\nabla s_i=\frac{1}{1+\varepsilon \rho \kappa_i} \mathbf{t}_i(\mathbf{s}, t), \quad\quad \text{for }i=1,2,\cdots,d-1.
\]

Any scalar function $ a(\mathbf{x},t) $ and vector field $ \mathbf{b}(\mathbf{x},t) $ can be expressed in the new coordinate system as $ a(\mathbf{x},t)=A(\mathbf{s},\rho,t) $ and $ \mathbf{b}(\mathbf{x},t)=\mathbf{B}(\mathbf{s},\rho,t) $ respectively. The following useful identities can be easily derived (see Ref.~\cite{garcke2023diffuse}):
\begin{equation}\label{coordinates_change}
	\begin{aligned}
		\partial_t a&
		%	=\partial_t A+\sum_{i=1}^{d-1} \partial_{s_i} A \partial_t s_i+\partial_\rho A \partial_t \rho
		=-\varepsilon^{-1}V\partial_\rho A+\partial_t^{\Gamma} A,\\
		\nabla a&
		%	=\nabla \rho \partial_\rho A+\sum_{i=1}^{d-1} \partial_{s_i} A \nabla s_i
		=\varepsilon^{-1} \mathbf{\nu}\partial_\rho A+\sum_{i=1}^{d-1}\frac{1}{1+\varepsilon\rho\kappa_i}\mathbf{t}_i\partial_{s_i}A=\varepsilon^{-1} \mathbf{\nu}\partial_\rho A+\nabla_{\mathbf{s}} A+\mathcal{O}(\varepsilon), \\
		\nabla \cdot \mathbf{b}&
		%	=\nabla \rho \cdot \partial_\rho \mathbf{B}+\sum_{i=1}^{d-1} \partial_{s_i} \mathbf{B} \cdot \nabla s_i
		=\varepsilon^{-1} \mathbf{\nu} \cdot \partial_\rho \mathbf{B}+\sum_{i=1}^{d-1}\frac{1}{1+\varepsilon\rho\kappa_i}\mathbf{t}_i\cdot\partial_{s_i}\mathbf{B} =\varepsilon^{-1} \mathbf{\nu} \cdot \partial_\rho \mathbf{B}+\nabla_{\mathbf{s}} \cdot \mathbf{B}+O(\varepsilon)
	\end{aligned}
\end{equation}
where $ V $ is the velocity of $\Gamma_\varepsilon(t)$ in the normal direction, i.e., $V=-\partial_t d=-\varepsilon \partial_t \rho$, and
$$
\partial_t^{\Gamma} A:=\partial_t A+\sum_{i=1}^{d-1} \partial_{s_i} A \partial_t s_i,
$$
resembles a ``material derivative'' along the surface.
Moreover, $\nabla_{\mathbf{s}}=\sum_{i=1}^{d-1} \mathbf{t}_i \partial_{s_i}$ denotes the surface gradient operator on $\Gamma_\varepsilon(t)$.
%$ \nabla_s\cdot=\sum_{i=1}^{d-1} \mathbf{t}_i \cdot\partial_{s_i} $ is surface divergence.

Now, we are ready for the inner expansion. Let $ u(\mathbf{x},t) $ and $ \mu(\mathbf{x},t) $ be expanded as 
\begin{align*}
	&u(\mathbf{x},t)=U(\mathbf{s},\rho,t)=U_0+\varepsilon U_1+\cdots,\\
	&\mu(\mathbf{x},t)=\bar{\mu}(\mathbf{s},\rho,t)=\varepsilon^{-1}\bar{\mu}_{-1}+\bar{\mu}_0+\cdots.
\end{align*}
Using Taylor expansion, we have
\[
F(U)=F(U_0)+\varepsilon F'(U_0)U_1+\varepsilon^2 \left(F'(U_0)U_2+\frac{F''(U_0)}{2}U_1^2 \right)+O(\varepsilon^3).
\]
Similarly, we can expand $M(U)$ and $N(U)$. For simplicity, we denote the corresponding expansion by 
\[
\begin{aligned}
	M(U)=\bar{M}=M_0+\varepsilon M_1+\varepsilon^2 M_2+\cdots,\quad
	N(U)=\bar{N}=N_0+\varepsilon N_1+\varepsilon^2 N_2+\cdots,
\end{aligned}	
\]
where $M_0=M(U_0)$ and $N_0=N(U_0)=1/Q'(U_0)$.

Using \eqref{coordinates_change}, \eqref{eq1} can be reformulated as: 
\begin{equation}\label{EQ1_left}
	\begin{aligned}
    &-\varepsilon^{-1} V\partial_\rho U+\partial_t^{\Gamma} U\\
		%&\varepsilon^{-1}N\nabla\cdot(M\nabla(N\mu))\\
        =&\varepsilon^{-3}\bar{N}\partial_\rho\left(\bar{M}\partial_{\rho}(\bar{N}\bar{\mu})\right)+\varepsilon^{-2}\bar{N}\sum_{i=1}^{d-1}\frac{1}{1+\varepsilon\rho\kappa_i}\mathbf{t}_i\cdot\partial_{s_i}\left(\bar{M}\mathbf{\nu}\partial_{\rho}(\bar{N}\bar{\mu})\right)\\
		&+\varepsilon^{-1}\bar{N}\sum_{i=1}^{d-1}\frac{1}{1+\varepsilon\rho\kappa_i}\mathbf{t}_i\cdot\partial_{s_i}\left(\bar{M}\sum_{j=1}^{d-1}\frac{1}{1+\varepsilon\rho\kappa_j}\mathbf{t}_j\partial_{\rho}(\bar{N}\bar{\mu})\right).
	\end{aligned}
\end{equation}
Next, from complex computation, we can expand \eqref{eq2} as
\begin{equation}
\begin{aligned}
\label{EQ2}
\bar{\mu}=&\varepsilon^{-1}(\gamma_0F'(U_0)-\gamma_0\partial_{\rho\rho}U_0)+\left(\gamma_0F''(U_0)U_1+\gamma_1\frac{F'(U_0)}{\partial_{\rho}U_0}-\gamma_0\partial_{\rho\rho}U_1-\partial_{\rho}\gamma_1\right)\\
&-\partial_{\rho}\left(\left(\mathbf{\nu}\cdot\nabla{\bm\xi}(\mathbf{\nu})\nabla_sU_0-\gamma_1\right)(\frac{1}{2}+\frac{F(U_0)}{(\partial_{\rho}U_0)^2})	\right)	\\
&-\nabla_{\mathbf{s}}\cdot\left(\gamma_0\partial_{\rho}U_0\mathbf{\nu}(\frac{1}{2}-\frac{F(U_0)}{(\partial_{\rho}U_0)^2})+\partial_{\rho}U_0{\bm\xi}(\mathbf{\nu})(\frac{1}{2}+\frac{F(U_0)}{(\partial_{\rho}U_0)^2})	\right)	+\mathcal{O}(\varepsilon).
\end{aligned}
\end{equation}
where $ \gamma_0=\hat{\gamma}(\mathbf{\nu}) $ and $\gamma_1={\bm\xi}(\mathbf{\nu})\cdot\nabla_sU_0$. 
Details for this derivation are given in Appendix \ref{expansion_second}. 
In fact,  the $O(1)$ terms in the brackets can be further simplified by using the matching results.
%the facts that $ U_0 $ is independent in $ \mathbf{s} $ and $ (\partial_{\rho}U_0)^2=2F(U_0) $. These are derived from the $O(1)$ inner solution that will be made clear in the next subsection. 
We will solve for the inner solutions at each order and perform matching with the outer solutions. 

\subsection{Matching}
Before we match outer solutions and inner solutions at each order, we match the flux $\mathbf{j}:=M\nabla(N\mu)$ between the inner region and the outer region, which helps us determine the inner solution and the outer solution. 

In the outer region, the flux can be expanded as
\begin{equation}\label{flux_out}
\mathbf{j}=\varepsilon^{w(l-k)}\frac{k!M^{(l)}(\pm1)}{l!Q^{(k+1)}(\pm1)}u_w^{lw}\nabla(\frac{1}{u_w^{kw}}\mu_0)+o(\varepsilon^{w(l-k)})
\end{equation}
where we have used the outer solutions $\mu_{-1}=0$ to eliminate the lowest order terms.
%, we know the leading order of the flux $\mathbf{j}$ in outer region is at the order $\mathcal{O}(\varepsilon)$.

For brevity of presentation, we define the following notations 
\[
% (n\mu)_k=\sum_{i\geqslant0,~j\geqslant-1,i+j=k} n_i\mu_j,\quad 
(N\bar{\mu})_k=\sum_{i\geqslant0,~j\geqslant-1,i+j=k} N_i\bar{\mu}_j
\]
In the inner region,  the normal component $ \mathbf{J}_\nu={\bm\nu}\cdot (M(U)\nabla(N(U)\bar{\mu})) $ can be expanded as:
\begin{equation}\label{flux_in}
	\begin{aligned}
		\mathbf{J}_\nu&=  \varepsilon^{-2}\left[M_0 \partial_\rho(N \bar{\mu})_{-1}\right] +\varepsilon^{-1}\left[M_1 \partial_\rho(N \bar{\mu})_{-1}+M_0 \partial_\rho(N \bar{\mu})_0\right] \\
		& +\left[M_2 \partial_\rho(N \bar{\mu})_{-1}+M_1 \partial_\rho(N \bar{\mu})_0+M_0 \partial_\rho(N \bar{\mu})_1\right] \\
		& +\varepsilon\left[M_3 \partial_\rho(N \bar{\mu})_{-1}+M_2 \partial_\rho(N \bar{\mu})_0+M_1 \partial_\rho(N \bar{\mu})_1+M_0 \partial_\rho(N \bar{\mu})_2\right]+\mathcal{O}\left(\varepsilon^2\right).
	\end{aligned}
\end{equation}
If $l\geqslant k$, the matching between the flux in the outer region and that in the inner region suggests that
\begin{align}\label{eq:matching_J}
	%		\mathbf{j}&=\varepsilon^{-1}\mathbf{j}_{_{-1}}+ \mathbf{j}_{_0}+\varepsilon \mathbf{j}_{_1}+\mathcal{O}(\varepsilon^2),\\
	\lim_{\rho\to\pm\infty}M_0\partial_\rho(N \bar{\mu})_{-1}=\lim_{\rho\to\pm\infty}\left[M_1 \partial_\rho(N \bar{\mu})_{-1}+M_0 \partial_\rho(N \bar{\mu})_0\right]=0.
    % =\lim_{\rho\to\pm\infty}\partial_\rho(N \bar{\mu})_{0}=0.
    %\mathbf{J}_\nu&=\varepsilon^{-2}\mathbf{J}_{\nu,-2}+\varepsilon^{-1}\mathbf{J}_{\nu,-1}+\mathbf{J}_{\nu,0}+\varepsilon \mathbf{J}_{\nu,1}+\mathcal{O}(\varepsilon^2).
\end{align}
Moreover, the leading-order matching condition in $u$ and $U$ gives
\begin{equation}\label{eq:matching_u}
\lim_{\rho\to\pm\infty}U_0=u_0=\pm1.
% \quad \lim_{\rho\to\pm\infty}U_1=u_1=0.
\end{equation}

\textit{Leading order.} Combining \eqref{EQ1_left} and \eqref{EQ2}, at the leading order of inner system, we have
\begin{numcases}{}
	0=N_0\partial_{\rho}\left(M_0\partial_{\rho}(N_0\bar{\mu}_{-1})\right), \label{iasymp_1}\\
	\bar{\mu}_{-1}=\gamma_0F'(U_0)-\gamma_0\partial_{\rho\rho}U_0\label{iasymp_1_2}.
\end{numcases}	
This system together with the matching conditions \eqref{eq:matching_J} and \eqref{eq:matching_u} leads to leading-order inner solutions
\begin{equation}\label{eq:in_leading}
U_0(\rho)=\sigma(\rho):=\tanh\bigg(\frac{\rho}{\sqrt{2}}\bigg),\quad
\bar{\mu}_{-1}=0.
\end{equation}
In fact, \eqref{iasymp_1} implies $M_0\partial_{\rho}(N_0\bar{\mu}_{-1})=A_0(\mathbf{s}, t)$ which does not depend on $\rho$.  
%From \eqref{flux_in}, we know that $ \mathbf{J}_{\nu,-2}=M_0\partial_{\rho}(N_0\bar{\mu}_{-1}) $ has to match $ \mathbf{j}\cdot\mathbf{\nu} $. Since the leading order of $ \mathbf{j}\cdot\mathbf{\nu}  $ is $\mathcal{O}(\varepsilon)$, we obtain 
The matching condition \eqref{eq:matching_J} indicates that $A_0=0$ and $ N_0\bar{\mu}_{-1}=B_0(\mathbf{s},t)$ independent of $\rho$. Then we multiply \eqref{iasymp_1_2} by $\partial_{\rho}U_0$ and integrate the result from $-\infty$ to $+\infty$ in $\rho$. The term on the left becomes
\begin{equation*}		 			\int_{-\infty}^{+\infty}\bar{\mu}_{-1}\partial_{\rho}U_0\drho=\int_{-\infty}^{+\infty}N_0\bar{\mu}_{-1}\frac{\partial_{\rho}U_0}{N_0}\drho=B_0\int_{-\infty}^{+\infty}Q'(U_0)\partial_{\rho}U_0\drho=2B_0,
\end{equation*}
where we have used the definition of $N_0$, the matching condition \eqref{eq:matching_u} and the property that $Q(\pm1)=\pm1$.
Direct calculation shows the term on the right vanishes:
\begin{equation*}	 			\int_{-\infty}^{+\infty}\gamma_0\partial_{\rho}U_0\left(F'(U_0)-\partial_{\rho\rho}U_0	\right)\drho=\gamma_0\left(F(U_0)-\frac{(\partial_{\rho}U_0)^2}{2}\right)\bigg|_{\rho=-\infty}^{+\infty}=0.
\end{equation*}
%The last equality of \eqref{right_hand_leading} is from matching condition: $ \lim\limits_{\rho\to\pm\infty}U_0=u_0=\pm1 $ and $ \lim\limits_{\rho\to\pm\infty}\partial_{\rho}U_0=0 $. 
As a result, we have $ B_0=0 $ and thus $ \bar{\mu}_{-1}=0 $.
Since $\gamma_0\neq0$, \eqref{iasymp_1_2} becomes
\begin{equation}\label{eq_tanh}
	%\left\{\begin{array}{l}
		\partial_{\rho\rho}U_0-F'(U_0)=0.
		%U_0(0,\mathbf{s},t)=0.
	%\end{array}\right.
\end{equation}
Solving this equation by aligning the zero-contour of $U_0$ with the interface $\Gamma_{\varepsilon}(t)$, i.e., $U_0(0,\mathbf{s},t)=0$, we can arrive at \eqref{eq:in_leading}. This also suggests that $U_0$ does not explicitly depend on $(\mathbf{s},t)$ and
\begin{equation}
\label{eq_tanh1}
(\partial_{\rho}U_0)^2=2F(U_0).
\end{equation}

\textit{First order.}
The first-order equation of the inner system is
\begin{numcases}{}
	0=N_0\partial_{\rho}\left(M_0\partial_{\rho}(N_0\bar{\mu}_0)\right), \label{iasymp_2}\\
	\bar{\mu}_0=\gamma_0F''(\sigma(\rho))U_1-\gamma_0\partial_{\rho\rho}U_1-\nabla_{\mathbf{s}}\cdot(\sigma'(\rho){\bm\xi}(\mathbf{\nu})).
	\label{iasymp_2_2}
\end{numcases}	
where we have used the fact that $U_0$ is independent of $\mathbf{s}$ (so that $\gamma_1=0$ according to \eqref{EQ2} and the definition therein) and \eqref{eq_tanh1}.
By the same argument as in the leading order,
\eqref{iasymp_2} and \eqref{eq:matching_J} implies that 
% $M_0\partial_{\rho}(N_0\bar{\mu}_0)=\text{constant }A_1 $ in $\rho$. From flux matching, which is similar to the leading order, we have $A_1=0$ that implies 
$N_0\bar{\mu}_0=B_1(\mathbf{s},t)$ independent of $\rho$. 
%Next, we will compute $ U_1 $ and $ \bar{\mu}_1 $. 
Let $\mathcal{L}:=\partial_{\rho\rho}-F''(\sigma(\rho))\mathbf{1}$ be a self-adjoint operator. \eqref{eq_tanh}  implies that $\mathcal{L}$ has a nontrivial null space spanned by $\sigma'(\rho)$.
Then a multiplication  by $U_0'(\rho)=\sigma'(\rho)$ on both sides of \eqref{iasymp_2_2} followed by an integration over $(-\infty,+\infty)$ leads to
% \begin{equation}\label{left_hand}		 			\int_{-\infty}^{+\infty}\bar{\mu}_0\sigma'(\rho)\drho=2B_1,
%=\int_{-\infty}^{+\infty}N_0\bar{\mu}_0\frac{q'(\rho)}{N_0}\drho=B_1\int_{-\infty}^{+\infty}\frac{\partial_{\rho}U_0}{N(U_0)}\drho
%\end{equation}
\begin{equation*}\label{right_hand}
\begin{aligned}
2B_1=&\int_{-\infty}^{+\infty}N_0\bar{\mu}_0\mathrm{d}(Q(\sigma))=\int_{-\infty}^{+\infty}\bar{\mu}_0\sigma'(\rho)\drho\\
=&\int_{-\infty}^{+\infty}\left(-\gamma_0\mathcal{L}U_1-\nabla_s\cdot(\sigma'(\rho){\bm\xi}(\mathbf{\nu})) \right)\sigma'(\rho)\drho
		%=&-\gamma_0\int_{-\infty}^{+\infty}\mathcal{L}(q'(\rho))\cdot U_1\drho-\nabla_s\cdot({\bm\xi}(\mathbf{\nu}))\int_{-\infty}^{+\infty}(q'(\rho))^2\drho\\
		=-\frac{2\sqrt{2}}{3}\nabla_{\mathbf{s}}\cdot{\bm\xi}(\mathbf{\nu}).
	\end{aligned}
\end{equation*}
%equalities use the integration by parts formula and properties: $ q(\rho)=\tanh(\frac{\rho}{\sqrt{2}}) $. 
As a result,
\begin{equation}
\label{eq_n0mu0}
     N_0\bar{\mu}_0=-\frac{\sqrt{2}}{3}\nabla_{\mathbf{s}}\cdot{\bm\xi}(\mathbf{\nu}).
\end{equation}
Substituting it into \eqref{iasymp_2_2}, we have the equation for $U_1$:
\begin{equation}\label{eq_u1}
	\left\{\begin{array}{l}
		\mathcal{L}U_1=A(\rho)\sigma'(\rho),\\
        %A(\rho)\sigma'(\rho),\\
		U_1(\rho=0,s,t)=0.
	\end{array}\right.
\end{equation}
where $A(\rho):=\Big(\frac{\sqrt{2}Q'(\sigma)}{3\sigma'(\rho)}-1\Big)\nabla_s\cdot{\bm\xi}(\mathbf{\nu})$. Due to \eqref{eq_tanh1}, we obtain that $
A(\rho)=\Big(\frac{(\sqrt{2})^{k+1}(\sigma'(\rho))^{k-1}}{3\int_0^1(1-z^2)^k\mathrm{d}z}-1\Big)\nabla_s\cdot{\bm\xi}(\mathbf{\nu})$ is a smooth bounded function of $\rho$ with finite limits as $\rho\to\pm\infty$ when $k\geqslant1$. In particular, when $k=1$, $A(\rho)\equiv0$. This implies that the source term of \eqref{eq_u1} has a degeneracy of at least the order of $\sigma'(\rho)$ at infinity. Since $\mathcal{L}(\sigma'(\rho))=0$, the solvability condition of \eqref{eq_u1} gives that
\begin{equation}
\label{eq:solvability}
	\int_{-\infty}^{+\infty}A(\rho)(\sigma'(\rho))^2\mathrm{d}\rho=0,
\end{equation}
which is satisfied by all $Q$ with $Q(\pm1)=\pm1$.

 Now we are ready to solve \eqref{eq_u1} for a smooth bounded solution that can be matched with the outer solution $u_1$ when $\rho\rightarrow\pm\infty$. Since both $\mathcal{L}$ and the source term of \eqref{eq_u1} are even in $\rho$, the solution $U_1$ has to be an even function. Multiplying \eqref{eq_u1} by $\sigma'(\rho)$ and integrating the result over $(0,\rho)$, we obtain
\begin{equation*}
    \partial_\rho U_1 \sigma'(\rho)-U_1\sigma''(\rho)=\int_{0}^\rho A(\zeta)(\sigma'(\zeta))^2\mathrm{d}\zeta,
\end{equation*}
where \eqref{eq_tanh} and the solvability condition \eqref{eq:solvability} have been used.
Then a multiplication by $1/(\sigma'(\rho))^2$ followed by an integration over $(0,\rho)$ gives 
\[
U_1(\rho)=\sigma'(\rho)\int_0^{\rho}\left(\frac{1}{(\sigma'(\eta))^2}\int^{\eta}_{0}A(\zeta)(\sigma'(\zeta))^2\mathrm{d}\zeta\right)\mathrm{d}\eta.
\]

% Using the method of variation of constants to solve \eqref{eq_u1}, we obtain that
% \[
% U_1(\rho)=\sigma'(\rho)\int_0^{\rho}\left(\frac{1}{(\sigma'(\eta))^2}\int^{\eta}_{-\infty}A(\zeta)\sigma'(\zeta)\mathrm{d}\zeta\right)\mathrm{d}\eta.
% \]
In order to match with the outer solution, we also need to compute the limits of $U_1$ when $\rho\to\pm\infty$. In fact, by applying L'Hopital's rule, we have
%We claim that there exists a constant $C$ such that $|U_1|\leqslant C$ because $U_1$ is continuous and 
\[
\begin{aligned}   \lim\limits_{\rho\to\pm\infty}U_1&=\lim\limits_{\rho\to\pm\infty}\frac{\frac{1}{(\sigma'(\rho))^2}\int^{\rho}_{-\infty}A(\zeta)(\sigma'(\zeta))^2\mathrm{d}\zeta}{-\frac{\sigma''(\rho)}{(\sigma'(\rho))^2}}=-\lim\limits_{\rho\to\pm\infty}\frac{A(\rho)(\sigma'(\rho))^2}{\sigma'''(\rho)}\\
%=\lim\limits_{\rho\to\pm\infty}\frac{\int^{\rho}_{-\infty}A(\zeta)(\sigma'(\zeta))^2\mathrm{d}\zeta}{-\sigma''(\rho)}\\
	&=-\frac{1}{\sqrt{2}}\lim\limits_{\rho\to\infty}\frac{A(\rho)\text{sech}^2\left(\frac{\rho}{\sqrt{2}}\right)}{2\tanh^2\left(\frac{\rho}{\sqrt{2}}\right)-\text{sech}^2\left(\frac{\rho}{\sqrt{2}}\right)}=0,
\end{aligned}
\]
where in the last equality we have used the fact that
%is due to $\lim\limits_{\rho\to\infty}\text{sech}^2\left(\frac{\rho}{\sqrt{2}}\right)=0$ and 
$A(\rho)$ is bounded. This implies that $U_1$ is bounded.
%In fact, from \eqref{eq_tanh}, we can get  $\sigma'(\rho)=\sqrt{2F(\sigma)}=\frac{1-\sigma^2}{\sqrt{2}}$, which implies that
% \begin{equation*}
% 	%	\label{eq_n}
% 	Q'(\sigma)=\frac{(1-\sigma^2)^k}{\int_0^1(1-z^2)^k\mathrm{d}z}=\frac{(\sqrt{2})^k(\sigma'(\rho))^k}{\int_0^1(1-z^2)^k\mathrm{d}z}.
% \end{equation*} 
% Since $N_0\bar{\mu}_0=B_1$, we have 
%We can observe that $A(\rho)$ is bounded because we have assumed that the interface $\partial \Omega^+_{\varepsilon}(t)$ possesses $\mathcal{C}^{2}$ regularity at all times.
Using the matching between $u_1$ and $U_1$, the outer system \eqref{out_second} yields that $(u_1, \mu_{0})=(0,0)$.

% This limitation also implies that the inner solution $U_1$ can match with the outer solution $u_1$, i.e., $\lim\limits_{\rho\to\infty}U_1=0=u_1$. 

To conclude, the outer solution is $u_{out}=\pm1+\mathcal{O}(\varepsilon^2)$ and the inner solution is $U=\sigma(\rho)+\varepsilon U_1+\mathcal{O}(\varepsilon^2)$.
Combining them, we obtain the second-order approximation \eqref{second_u},
\[
u_{\varepsilon} = u_{unif,0}+\varepsilon u_{unif,1}+\mathcal{O}(\varepsilon^2),
\]
where $u_{unif,0}=\sigma\Big( \frac{\text{dist}\left(x, \partial \Omega_\varepsilon^+(t)\right)}{\varepsilon}\Big) $, and there exists a constant $C$ such that $ |u_{unif,1}|\leqslant C$. In particular, when $k=1$, we have $U_1=0$ and thus $u_{unif,1}=0$.

\textit{Second order.}
Using the fact that $ \bar{\mu}_{-1}=0 $ and $ N_0\bar{\mu}_0 $ is independent of $\rho$, the second order of the inner system is
\[
0=N_0\partial_{\rho}\left(M_0\partial_{\rho}(N_0\bar{\mu}_1+N_1\bar{\mu}_0)\right),
\]
By the same argument as in the leading order, it follows that $N_0\bar{\mu}_1+N_1\bar{\mu}_0=B_2(\mathbf{s},t) $ independent of $\rho$.

\textit{Third order.}
% Since $ \bar{\mu}_{-1}=0 $, and $ N_0\bar{\mu}_0,~N_0\bar{\mu}_1+N_1\bar{\mu}_0 $ are constants in $\rho$. For the inner system at the third order, we have
By the same argument, the third order of the inner system is 
\begin{equation*}\label{iasymp_3}
	-V\partial_{\rho}U_0=N_0\partial_{\rho}\left(M_0\partial_{\rho}(N_2\bar{\mu}_0+N_1\bar{\mu}_1+N_0\bar{\mu}_2)	\right)+N_0\nabla_{\mathbf{s}}\cdot(M_0\nabla_s(N_0\bar{\mu}_0)).
\end{equation*} 
Multiplying it by $1/N_0$ on both sides, and integrating the result
over $(-\infty,\infty)$, we obtain that 
\begin{equation*}\label{eq_left_v}
\begin{aligned}
	-2V=&-V\int_{-\infty}^{+\infty}Q'(U_0)\partial_{\rho}U_0\mathrm{d}\rho\\
=&M_0\partial_{\rho}(N_2\bar{\mu}_0+N_1\bar{\mu}_1+N_0\bar{\mu}_2)	\bigg|_{-\infty}^{+\infty}+\int_{-\infty}^{+\infty}\nabla_{\mathbf{s}}\cdot(M_0\nabla_{\mathbf{s}}(N_0\bar{\mu}_0))\drho\\
		=&-\frac{\sqrt{2}}{3}\Delta_{\mathbf{s}}(\nabla_{\mathbf{s}}\cdot({\bm\xi}(\mathbf{\nu})))\int_{-\infty}^{+\infty}M(U_0)\drho,
	\end{aligned}
\end{equation*}
where  \eqref{eq_n0mu0} has been used. %the second right-hand side term  becomes
% \begin{equation*}\label{eq_right_v2}
% 	\begin{aligned}
% 		&\int_{-\infty}^{+\infty}N_0\nabla_{\mathbf{s}}\cdot(M_0\nabla_{\mathbf{s}}(N_0\bar{\mu}_0))\frac{1}{N_0}\drho\\
% 		=&-\frac{\sqrt{2}}{3}\Delta_{\mathbf{s}}(\nabla_{\mathbf{s}}\cdot({\bm\xi}(\mathbf{\nu})))\int_{-\infty}^{+\infty}M(U_0)\drho.
% 	\end{aligned}
% \end{equation*}
This immediately leads to the law of anisotropic surface diffusion \eqref{pure_surfacediffusion} at the leading order, with the normal velocity given by
\[
V=C_l\Delta_{\mathbf{s}}(\nabla_{\mathbf{s}}\cdot({\bm\xi}(\mathbf{\nu}))),
\]
where
\[
C_l=\frac{\int_{-\infty}^{+\infty}M(U_0)\drho}{3\sqrt{2}}=\frac{1}{3}\int_{-\infty}^{+\infty}(1-U_0^2)^{l-1}\frac{(1-U_0^2)}{\sqrt{2}}\mathrm{d}\rho=\frac{2}{3}\int_{0}^{1}(1-z^2)^{l-1}\mathrm{d}z.
\]

\subsection{Second-order volume conservation}
Based on the second-order approximation \eqref{second_u}, we will show the second-order volume conservation property \eqref{second_volume}. We restate the conservation property \eqref{manifold_M} using the asymptotic solution $u_\varepsilon$:
\[
\int_\Omega Q(u_\varepsilon)\dx=\text{constant}.
\]
%Recall that $\sigma(\rho)=\tanh\left(\frac{\rho}{\sqrt{2}}\right)$.
%Given the second-order asymptotic expansion of \( u_\varepsilon \) in \eqref{second_u}, we can expand \( Q(u_\varepsilon) \) using the Taylor expansion as 
This can be expanded as
\begin{equation}
	\label{eq_taylor}
	\begin{aligned}
		&\int_{\Omega}Q(u_\varepsilon)\dx\\
		=&\int_\Omega Q\left(\sigma\left(\frac{\text{dist}(x,\partial\Omega_\varepsilon^+)}{\varepsilon}\right)\right)\dx+\varepsilon\int_\Omega Q'\left(\sigma\left(\frac{\text{dist}(x,\partial\Omega_\varepsilon^+)}{\varepsilon}\right)\right)u_{unif,1}\dx+\mathcal{O}(\varepsilon^2)\\
		:=&I+II+\mathcal{O}(\varepsilon^2),
	\end{aligned}
\end{equation}
where \( I \) and \( II \) represent the leading-order and first-order contributions, respectively.

We first introduce an area measure for the $\rho$-level-set: 
\[
h(\rho):=\mathcal{H}^{d-1}(\{x\in\Omega:\text{dist}(x,\partial\Omega_\varepsilon^+)=\rho\}),
\]
where \( \mathcal{H}^{d-1} \) denotes the \((d-1)\)-dimensional Hausdorff measure. If \( \partial\Omega_\varepsilon^+ \) is of class \( \mathcal{C}^2 \), then \( h(\rho) \) is well-defined for almost every \( \rho \).

For the term $I$, we employ the coarea formula and reformulate it as   
\begin{equation*}\label{eq_I}
	\begin{aligned}
		&I= \int_{-\infty}^{+\infty}Q\left(\sigma(\frac{\zeta}{\varepsilon})\right)h(\zeta)\mathrm{d}\zeta\\
		=&\int_{0}^{+\infty}h(\zeta)+\left(Q\left(\sigma(\frac{\zeta}{\varepsilon})\right)-1\right)h(\zeta)\mathrm{d}\zeta+\int_{-\infty}^{0}(-h(\zeta))+\left(Q\left(\sigma(\frac{\zeta}{\varepsilon})\right)+1\right)h(\zeta)\mathrm{d}\zeta\\
		=&|\Omega_{\varepsilon}^+(t)|-|\Omega_{\varepsilon}^-(t)|+\int_{0}^{+\infty}\left(Q\left(\sigma(\frac{\zeta}{\varepsilon})\right)-1\right)\left(h(\zeta)-h(-\zeta)\right)\mathrm{d}\zeta\\
		=&|\Omega_{\varepsilon}^+(t)|-|\Omega_{\varepsilon}^-(t)|+\varepsilon\int_{0}^{+\infty}\left(Q\left(\sigma(\zeta)\right)-1\right)\left(h(\varepsilon \zeta)-h(-\varepsilon \zeta)\right)\mathrm{d}\zeta,
	\end{aligned}
\end{equation*}
where the third equality holds because $ Q $ and $\sigma$ are odd functions.

We will show that  
\begin{equation}\label{eq_hausdorff_order}
	\int_{0}^{+\infty}\left(Q\left(\sigma(\zeta)\right)-1\right)\left(h(\varepsilon \zeta)-h(-\varepsilon \zeta)\right)\mathrm{d}\zeta=\mathcal{O}(\varepsilon).
\end{equation}
To this end, we split the integral into two parts based on the domain decomposition $(0,+\infty)=(0,|\log(\varepsilon)|]\bigcup(|\log(\varepsilon)|,+\infty))$.
% \begin{equation}
% 	\label{two_parts}
% 	\int_{0}^{+\infty}\left(Q\left(\sigma(\zeta)\right)-1\right)\left(h(\varepsilon \zeta)-h(-\varepsilon \zeta)\right)\mathrm{d}\zeta=\int_{0}^{|\log(\varepsilon)|}+\int^{+\infty}_{|\log(\varepsilon)|}.
% \end{equation}
We would show that the resulting two parts are both $\mathcal{O}(\varepsilon)$.

% For the first part. Let $\Gamma^s(t)=\left\{x \in \mathbb{R}^n, \text{dist}(x, \partial \Omega^+(t))=s\right\}$ for every $s \in \mathbb{R}^{+}$. As $\partial\Omega^+\in\mathcal{C}^2$ , we can assume that $\Gamma^s=\phi_s(\partial\Omega^+)$,  where $\phi_s$ is defined as $\phi_s(x)=x+s \nabla_x \text{dist}(x,t)$. Indeed, for $s$ sufficiently small, $\left(\phi_s\right)_s$ a family of diffeomorphisms with initial velocity given by $\nabla_x \text{dist}$.
By the first variation formula of the perimeter\cite{maggi2012sets}, we have
\[
h(\zeta)=h(0)+\zeta\int_{\partial\Omega_\varepsilon^+}\text{div}_{\partial\Omega_\varepsilon^+}\nabla \text{dist}(x)\mathrm{d}\mathcal{H}^{d-1}+\mathcal{O}(\zeta^2),
\]
which is simply rewritten as $ h(\zeta)=h(0)+\zeta h'(0)+\mathcal{O}(\zeta^2) $. Similarly, we have $ h(-\zeta)=h(0)-\zeta h'(0)+\mathcal{O}(\zeta^2) $. Combining them, we obtain that
\[
h(\zeta)-h(-\zeta)=2 \zeta h^{\prime}(0)+\mathcal{O}\left(\zeta^2\right).
\]
Now, for all $\varepsilon>0$ and $\zeta \in[ 0, |\log \varepsilon|]$, we have $\zeta \varepsilon \rightarrow 0$ as $\varepsilon \rightarrow 0$, and the following estimate holds uniformly as $\varepsilon \rightarrow 0$:
\[
h(\varepsilon \zeta)-h(-\varepsilon \zeta)=2 \zeta \varepsilon h^{\prime}(0)+\mathcal{O}\left(\varepsilon^2 \zeta^2\right),
\]
where the upper bound of the term $\mathcal{O}\left(\varepsilon^2 \zeta^2\right)$ is $C_1\varepsilon^2 \zeta^2$ with $C_1$ independent of $\zeta$.
Since $ 1-Q(u) $ is a polynomial function which vanishes at $ u=1 $ and $ \sigma(\zeta) $ converges to $ 1 $ exponentially as $ \zeta\to+\infty $, the moment $ \int_0^{+\infty}\zeta^k(1-Q(\sigma(\zeta)))\mathrm{d}\zeta $ is finite for all $k\geqslant1$. Hence,

\begin{equation}\label{eq_small}
	\begin{aligned}
		&\left|\int_{0}^{|\log(\varepsilon)|}\left(Q\left(\sigma(\zeta)\right)-1\right)\left(h(\varepsilon \zeta)-h(-\varepsilon \zeta)\right) \mathrm{d}\zeta\right|\\
		\leqslant&\varepsilon\int_{0}^{+\infty}\left(1-Q\left(\sigma(\zeta)\right)\right)\Big(2\zeta |h'(0)|+C_1\varepsilon \zeta^2\Big)\mathrm{d}\zeta =\mathcal{O}(\varepsilon).
	\end{aligned}
\end{equation}

%For the second part of \eqref{two_parts}.
Since $ h(\zeta) $ is bounded when $ \zeta<0$\cite{ambrosio2000geometric},
%We can assume that $ 
%(the level set inside a smooth and close surface $h(0)$ has bounded length).
we can immediately derive that 
\begin{equation}\label{eq_big1}
	\left|\int^{+\infty}_{|\log(\varepsilon)|}h(-\varepsilon \zeta) \left(Q\left(\sigma(\zeta)\right)-1\right) \mathrm{d}\zeta\right|\leqslant C\int^{+\infty}_{|\log(\varepsilon)|} \left|Q\left(\sigma(\zeta)\right)-1\right| \mathrm{d}\zeta=\mathcal{O}(\varepsilon).
\end{equation}
From the properties of the Hausdorff measure\cite{ambrosio2000geometric}, in particular its invariance by isometry and its dilation properties, if $ \zeta\to+\infty $, we have 
\[
h(\zeta)=\mathcal{O}(\zeta^{d-1}),
\]
which leads to
\begin{equation}\label{eq_big2}
	\left|\int^{+\infty}_{|\log(\varepsilon)|}h(\varepsilon \zeta)\left(Q\left(\sigma(\zeta)\right)-1\right) \mathrm{d}\zeta\right|\leqslant C\varepsilon^{d-1}\int^{+\infty}_{0}\zeta^{d-1}\left|Q\left(\sigma(\zeta)\right)-1\right| \mathrm{d}\zeta=\mathcal{O}(\varepsilon^{d-1}).
\end{equation}
Combing \eqref{eq_small}, \eqref{eq_big1} and \eqref{eq_big2}, we have proved \eqref{eq_hausdorff_order}.

For the term $II$ in \eqref{eq_taylor}, we have
\[
\begin{aligned}
	|II|&\leqslant\varepsilon\|u_{unif,1}\|_\infty\int_{-\infty}^{+\infty} \left|Q'\left(\sigma\left(\frac{\zeta}{\varepsilon}\right)\right)h(\zeta)\right|\mathrm{d}\zeta \\
	&=\varepsilon^2\|u_{unif,1}\|_\infty\int_{-\infty}^{+\infty} \left|Q'\left(\sigma(\zeta)\right)h(\varepsilon \zeta)\right|\mathrm{d}\zeta.
\end{aligned}
\]
We will show that $\int_{-\infty}^{+\infty}\left| Q'\left(\sigma(\zeta)\right)h(\varepsilon \zeta)\right|\mathrm{d}\zeta =\mathcal{O}(1)$ by splitting the integral into three parts based on the domain decomposition $(-\infty,+\infty)=[-|\log(\varepsilon)|,|\log(\varepsilon)|]\bigcup(-\infty,-|\log(\varepsilon)|)\bigcup(|\log(\varepsilon)|,+\infty)$.

% Splitting the integral into three parts based on the range of \( \zeta \), we have
% \begin{equation}\label{three_parts}
% 	\int_{-\infty}^{+\infty}\left| Q'\left(\sigma(\zeta)\right)h(\varepsilon \zeta)\right|\mathrm{d}\zeta=	\int_{-|\log \varepsilon|}^{|\log \varepsilon|}
% 	+\int_{-\infty}^{-|\log \varepsilon|}+\int_{|\log \varepsilon|}^{\infty}.
% \end{equation}
For the integral over $[-|\log(\varepsilon)|,|\log(\varepsilon)|]$, since $\varepsilon|\log \varepsilon|\to0$, we also have the following estimate uniformly as $\varepsilon \rightarrow 0$:
\[
h(\varepsilon \zeta)=h(0) +\mathcal{O}\left(\varepsilon \zeta\right),
\]
where the upper bound of the term $\mathcal{O}\left(\varepsilon \zeta\right)$ is $C_2\varepsilon \zeta$ with $C_2$ independent of $\zeta$. 
By a similar argument as in deriving \eqref{eq_small}, we have
\begin{equation}\label{three_I}
	\begin{aligned}
		&\int_{-|\log \varepsilon|}^{|\log \varepsilon|}\left| Q'\left(\sigma(\zeta)\right)h(\varepsilon \zeta)\right|\mathrm{d}\zeta\\
        \leqslant &\int_{-|\log \varepsilon|}^{|\log \varepsilon|}|Q'\left(\sigma(\zeta)\right)|(h(0)+C_2\varepsilon\zeta)\mathrm{d}\zeta\\
		%\leqslant& \int_{-|\log \varepsilon|}^{|\log \varepsilon|}|Q'\left(\sigma(\zeta)\right)h(0)|\mathrm{d}\zeta+O(\varepsilon|\log\varepsilon|)\\
        \leqslant& h(0)\int_{-\infty}^{\infty}|Q'\left(\sigma(\zeta)\right)|\mathrm{d}\zeta+\mathcal{O}(\varepsilon)=\mathcal{O}(1).
	\end{aligned}
\end{equation}
Using similar arguments as in \eqref{eq_big1} and \eqref{eq_big2}, we can show that the other two integrals are $o(1)$ quantities:
\begin{equation}\label{three_II}
	\int_{-\infty}^{-|\log \varepsilon|}\left| Q'\left(\sigma(\zeta)\right)h(\varepsilon \zeta)\right|\mathrm{d}\zeta+\int_{|\log \varepsilon|}^{\infty}\left| Q'\left(\sigma(\zeta)\right)h(\varepsilon s)\right|\mathrm{d}\zeta=o(1).
\end{equation}
Combing \eqref{three_I} and \eqref{three_II}, we obtain  $\int_{-\infty}^{+\infty}\left| Q'\left(\sigma(\zeta)\right)h(\varepsilon \zeta)\right|\mathrm{d}\zeta =\mathcal{O}(1)$ that implies
\[
|II|\leqslant \varepsilon^2\|u_{unif,1}\|_\infty\int_{-\infty}^{+\infty} \left|Q'\left(\sigma(\zeta)\right)h(\varepsilon \zeta)\right|\mathrm{d}\zeta=\mathcal{O}(\varepsilon^2).
\]

Collecting the results for \( I \) and \( II \), we have 
\[
\text{constant}=\int_\Omega Q(u_\varepsilon)\dx=|\Omega_{\varepsilon}^+(t)|-|\Omega_{\varepsilon}^-(t)|+\mathcal{O}(\varepsilon^2).
\]
Using the fact that \( |\Omega_\varepsilon^+(t)| + |\Omega_\varepsilon^-(t)| = |\Omega|=\text{constant} \), we conclude the second-order volume conservation \eqref{second_volume}.
% that  
% \[
% |\Omega_\varepsilon^+(t)| = |\Omega_\varepsilon^+(0)| + \mathcal{O}(\varepsilon^2).
% \]  
%establishing the second-order real volume conservation.

\section{Numerical results}\label{chapter4}
Due to the variational structure of the ACH-IC, the corresponding efficient numerical schemes can be designed easily by using many existing techniques, for example, the stabilized semi-implicit schemes\cite{shen2018stabilized}, invariant energy quadratization (IEQ) approach\cite{yang2017numerical}, scalar auxiliary variable (SAV) approach\cite{shen2019new}, etc. In this paper, we adopt a linear numerical scheme based on the stabilized-invariant energy quadratization (S-IEQ) approach\cite{xu2019efficient,yang2021efficient} to demonstrate the approximation performance of the ACH-IC for anisotropic surface diffusion. A more detailed discussion is provided in Appendix \ref{numerical_scheme} for reference.

It is worth noting that the mobility \( M(u)=(1-u^2)^l \) in our model includes a parameter \( l \), which directly affects the associated normal velocity, as given in \eqref{pure_surfacediffusion}. To facilitate meaningful comparisons and ensure consistency, we apply a temporal scaling transformation to normalize the interface diffusion coefficient to 1. Specifically, we rescale the time variable \( t \) by \( t' = t/C_l \), such that the diffusion speed in the rescaled time domain becomes independent of \( l \). If not otherwise specified, we set \( M(u) = (1 - u^2)^2 \), and in this case, $C_l=4/9$. Moreover, we set $Q(u)=\frac{\int_0^u(1-z^2)^k\mathrm{d}z}{\int_0^1(1-z^2)^k\mathrm{d}z}$ with $k=1$.

The choice of anisotropic energy density $ \gamma(\mathbf{n}) $ depends on the specific material properties. Here, we present two commonly used forms:
\begin{enumerate}
	\item Four-fold form\cite{wise2007solving}:
	\begin{equation}\label{four_fold}
		\gamma(\mathbf{n}) = 1 + \alpha( 4\sum_{i=1}^{d} n_i^4-3),
	\end{equation}
	where $\alpha$ is a constant that governs the strength of anisotropy. 	
	If $\alpha<1/15$, this anisotropic function is weakly anisotropic. Otherwise, it is strongly anisotropic.
	
	\item Riemannian metric form (also called BGN anisotropies)\cite{barrett2008numerical}:
	\begin{equation}\label{Riemannian_metric_form}
		\gamma(\mathbf{n}) = \sqrt{\mathbf{n}^T R \mathbf{n}},
	\end{equation}
	where $R$ is a symmetric positive definite matrix. This form encompasses ellipsoidal anisotropy\cite{zhao2020parametric} and certain Riemannian metric-based anisotropy\cite{jiang2019sharp}. The matrix $R$ encodes the specific anisotropic characteristics of the material.
\end{enumerate}

\subsection{Accuracy test for ellipsoidal anisotropic density} 
The computational domain is a square region $\Omega=[0,1]^2$, and the initial condition is represented by a circle centered at $(0.5,0.5)$ defined by 
\[
u(x,y,t=0)=-\tanh\left(\frac{\sqrt{(x-0.5)^2+(y-0.5)^2}-0.3}{5\sqrt{2}\times10^{-3}}\right).
\]
The interface $ \{x: u(x)=0\} $ is determined by interpolation, and we call the area of the domain enclosed by this interface as real volume. Table \ref{tab2:volume} presents the results of the error in the real volume and its convergence rate obtained by the ACH-IC \eqref{eq_NMN} and the ACH \eqref{eq_M} for various thickness parameter $\varepsilon$ and four different finial time $ T $ with time step size $\delta t=10^{-6}$. In these experiments, we considered the ellipsoidal anisotropy, defined as $\gamma(\mathbf{n})=\sqrt{2n_1^2+n_2^2}$. The theoretical prediction for the equilibrium shape under this anisotropy is a self-similar ellipsoid given by $\frac{x^2}{2}+y^2=1$ (as discussed in Ref.~\cite{zhao2020parametric}). 
It is evident from Table \ref{tab2:volume} that the ACH and the ACH-IC exhibit first-order and second-order real volume conservation, respectively. This observation aligns with the theoretical expectations. We can further observe that the error in the real volume of the ACH-IC is an order of magnitude smaller than that of the ACH model. 

\begin{table}[htb]
	%	\vspace{-1cm}
	\centering
	\fontsize{7}{10}\selectfont
	\caption{Convergence rate of real volume conservation for the ACH  and the ACH-IC, as $\varepsilon$ is gradually reduced, under the anisotropy $\gamma(\mathbf{n})=\sqrt{2n_1^2+n_2^2}$.}
	\label{tab2:volume}
	\begin{tabular}{|c|c|c|c|c|c|c|c|c|}
		\hline
		\multirow{3}{*}{$\varepsilon$}&
		\multicolumn{4}{c|}{Final time: $ T=1\times10^{-4} $}&\multicolumn{4}{c|}{Final time: $ T=5\times10^{-4} $}\cr\cline{2-9}&
		\multicolumn{2}{c|}{ACH}&\multicolumn{2}{c|}{ACH-IC}&\multicolumn{2}{c|}{ACH}&\multicolumn{2}{c|}{ACH-IC}\cr\cline{2-5}\cline{6-9}
		&error&order&error&order&error&order&error&order\cr
		\hline
		\hline
		0.08
		&2.32$\times10^{-2}$&$\backslash$&1.58$\times10^{-2}$&$\backslash$
		&4.21$\times10^{-2}$&$\backslash$&1.70$\times10^{-2}$&$\backslash$
		\cr\hline
		0.04
		&1.10$\times10^{-2}$&1.08&3.38$\times10^{-3}$&2.22
		&1.48$\times10^{-2}$&1.50&3.34$\times10^{-3}$&2.35
		\cr\hline
		0.02
		&3.81$\times10^{-3}$&1.52&8.09$\times10^{-4}$&2.06
		&4.99$\times10^{-3}$&1.57&7.77$\times10^{-4}$&2.10
		\cr\hline
		0.01
		&1.22$\times10^{-3}$&1.64&2.05$\times10^{-4}$&1.98
		&1.70$\times10^{-3}$&1.55&1.99$\times10^{-4}$&1.97
		\cr\hline
		\hline
		\multirow{3}{*}{$\varepsilon$}&
		\multicolumn{4}{c|}{Final time: $ T=1\times10^{-3} $}&\multicolumn{4}{c|}{Final time: $ T=2\times10^{-3} $}\cr\cline{2-9}&
		\multicolumn{2}{c|}{ACH}&\multicolumn{2}{c|}{ACH-IC}&\multicolumn{2}{c|}{ACH}&\multicolumn{2}{c|}{ACH-IC}\cr\cline{2-9}
		&error&order&error&order&error&order&error&order\cr
		\hline
		\hline
		0.08
		&5.06$\times10^{-2}$&$\backslash$&1.87$\times10^{-2}$&$\backslash$
		&5.66$\times10^{-2}$&$\backslash$&2.06$\times10^{-2}$&$\backslash$
		\cr\hline
		0.04
		&1.70$\times10^{-2}$&1.57&3.34$\times10^{-3}$&2.49
		&1.91$\times10^{-2}$&1.56&3.34$\times10^{-3}$&2.62
		\cr\hline
		0.02
		&5.76$\times10^{-3}$&1.56&7.73$\times10^{-4}$&2.11
		&6.76$\times10^{-3}$&1.49&7.75$\times10^{-4}$&2.11
		\cr\hline
		0.01
		&2.03$\times10^{-3}$&1.50&1.95$\times10^{-4}$&1.99
		&2.46$\times10^{-3}$&1.45&1.93$\times10^{-4}$&2.00
		\cr\hline
	\end{tabular}
\end{table}
%\begin{figure}
%	\centering
%	\includegraphics[trim=2.2cm 0.3cm 2.2cm 0.4cm, clip,width=0.4\linewidth]{figures/evolution_fig}
%	\caption{Evolution of the circle under the NMN-ACH model with $\varepsilon=0.02$. The snapshots illustrate the gradual deformation of the circle at different time points. As time progresses, the circular shape transforms into an ellipse.}
%	\label{fig:evolutionfig}
%		\vspace{-0.6cm}
%\end{figure}
\subsection{Accuracy test for four-fold anisotropic density}
In this subsection, we test the convergence rate in $\varepsilon$ under the four-fold anisotropic density \eqref{four_fold} with $\alpha=0.05$ and $\alpha=0.2$. 
The computational parameters for weakly anisotropic density (i.e., $\alpha=0.05$) are the same as before. For strongly anisotropic density (i.e., $\alpha=0.2$), we set the regularized strength $\beta=\varepsilon^2$, and the stabilization parameters as $S_1=S_2=4$, $S_3=\beta$. Similar numerical results can be obtained for different time snapshots. 
Additionally, we show the results of the error in the real volume and its convergence rate for the equilibrium states obtained through the ACH and ACH-IC, as presented in Table \ref{tab:volume}. A clear second-order convergence is observed for the ACH-IC, while the convergence rate for the ACH is seriously below second order.  This observation aligns with our theoretical justification.
\begin{table}[htb]
	\centering
	\fontsize{7}{10}\selectfont
	\caption{Convergence rate of real volume conservation for the ACH  and the ACH-IC with the four-fold form anisotropy \eqref{four_fold} as $\varepsilon$ decreases for final time $ T=0.01 $.}
	\label{tab:volume}
	\begin{tabular}{|c|c|c|c|c|c|c|c|c|}
		\hline
		\multirow{3}{*}{$\varepsilon$}&
		\multicolumn{4}{c|}{Weakly anisotropic case ($\alpha=0.05$)}&\multicolumn{4}{c|}{Strongly anisotropic case ($\alpha=0.2$)}\cr\cline{2-9}&
		\multicolumn{2}{c|}{ACH}&\multicolumn{2}{c|}{ACH-IC}&\multicolumn{2}{c|}{ACH}&\multicolumn{2}{c|}{ACH-IC}\cr\cline{2-5}\cline{6-9}
		&error&order&error&order&error&order&error&order\cr
		\hline
		\hline
		0.08
		&6.23$\times10^{-2}$&$\backslash$&1.71$\times10^{-2}$&$\backslash$
		&4.08$\times10^{-2}$&$\backslash$&1.42$\times10^{-2}$&$\backslash$
		\cr\hline
		0.04
		&2.20$\times10^{-2}$&1.50&3.41$\times10^{-3}$&2.33
		&1.12$\times10^{-2}$&1.87&2.97$\times10^{-3}$&2.25
		\cr\hline
		0.02
		&8.95$\times10^{-3}$&1.29&8.05$\times10^{-4}$&2.08
		&3.53$\times10^{-3}$&1.66&7.07$\times10^{-4}$&2.07
		\cr\hline
		0.01
		&3.78$\times10^{-3}$&1.24&2.02$\times10^{-4}$&2.00
		&1.12$\times10^{-3}$&1.66&1.46$\times10^{-4}$&2.27
		\cr\hline
	\end{tabular}
\end{table}

%We also show the in Figure\ref{fig:evolutionfig}
%\begin{figure*}[htpb]
%	\centering
%	\begin{minipage}[c]{0.45\textwidth}
	%		\centering
	%		\includegraphics[trim=2.2cm 0.3cm 2.2cm 0.4cm, clip,width=1\linewidth]{figures/evolution_fig}
	%		\captionof{figure}{	\label{fig:evolutionfig}Boundedness of the sequence $\{\mathbf{z}^{(n+1)}\}_{n\geqslant0}  $ for two different initial step-size $ \tau_0=1 $ and $ 0.1 $.}
	%	\end{minipage}
%	~
%	\begin{minipage}[c]{0.45\textwidth}
	%		\centering
	%		\includegraphics[trim=0.5cm 0.3cm 0.5cm 0.3cm, clip,width=1\linewidth]{figures/evolution_fig}
	%		\captionof{figure}{\label{fig_relative_mass}Relative mass of the sequence $\{\mathbf{z}^{(n+1)}\}_{n\geqslant0}  $ for two different initial step-size $ \tau_0=1 $ and $ 0.1 $.}
	%	\end{minipage}
%\end{figure*}

Furthermore, we compare the equilibrium shapes obtained by the ACH and ACH-IC in Fig.~\ref{fig:equilibrium} and Fig.~\ref{strongly_equi}.  In these figures, we use the same color to label the cases with the same $\varepsilon$ and $\beta$, with solid lines representing equilibrium shapes obtained by the ACH-IC and dashed lines representing those obtained by the ACH.
Overall, all the numerical equilibrium shapes exhibit pyramid-like structures with corners, and the number of ``facets'' also demonstrates a four-fold geometric symmetry, which agrees with the exact equilibrium shape (determined by the Wulff construction\cite{wulff1901}) shown by a solid red line.  However, a close observation demonstrates that the ACH incurs larger errors compared to the ACH-IC, especially when $\varepsilon=0.08$. As depicted in the zoom-in plots, both models exhibit convergence as the thickness parameter $\varepsilon$ decreases, but the ACH-IC achieves a higher accuracy. Especially for the weakly anisotropic case (Fig.~\ref{fig:equilibrium}), the numerical result of the ACH-IC with $\varepsilon=0.04$ even outperforms that of the ACH with $\varepsilon=0.01$. Therefore, using a larger $\varepsilon$ for the ACH-IC is still capable of producing better simulation results. Since in phase-field simulations the mesh size is heavily dependent on $\varepsilon$, this improvement brought by the ACH-IC helps realize numerical simulation of surface diffusion with coarser mesh while retaining sufficient interface resolution and therefore leads to the reduction in computational cost.

\begin{figure}[htb]
	\centering
	\includegraphics[trim=0.7cm 0.2cm 0.3cm 0.5cm, clip,width=0.6\linewidth]{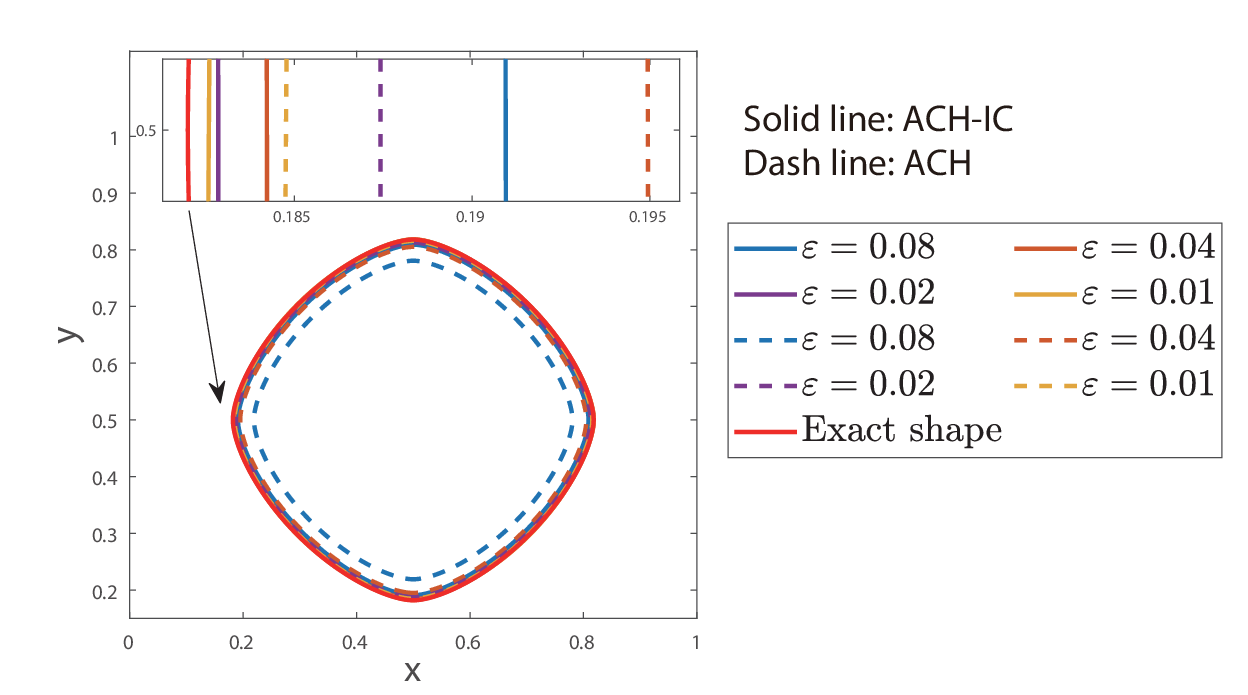}
	\caption{ Equilibrium shapes obtained by the ACH-IC (solid line) and ACH (dashed line) for weakly four-fold anisotropic density ($\alpha=0.05$) with thickness parameters $\varepsilon=0.08,~0.04,~0.02,~0.01$. The zoom-in plot nearby the left corner is shown, illustrating the convergence of both model and highlighting the superior accuracy of the ACH-IC.}
	\label{fig:equilibrium}
\end{figure}

\begin{figure}[htb]
	\centering
	\includegraphics[trim=0.5cm 0.2cm 0.3cm 0.5cm, clip,width=0.6\linewidth]{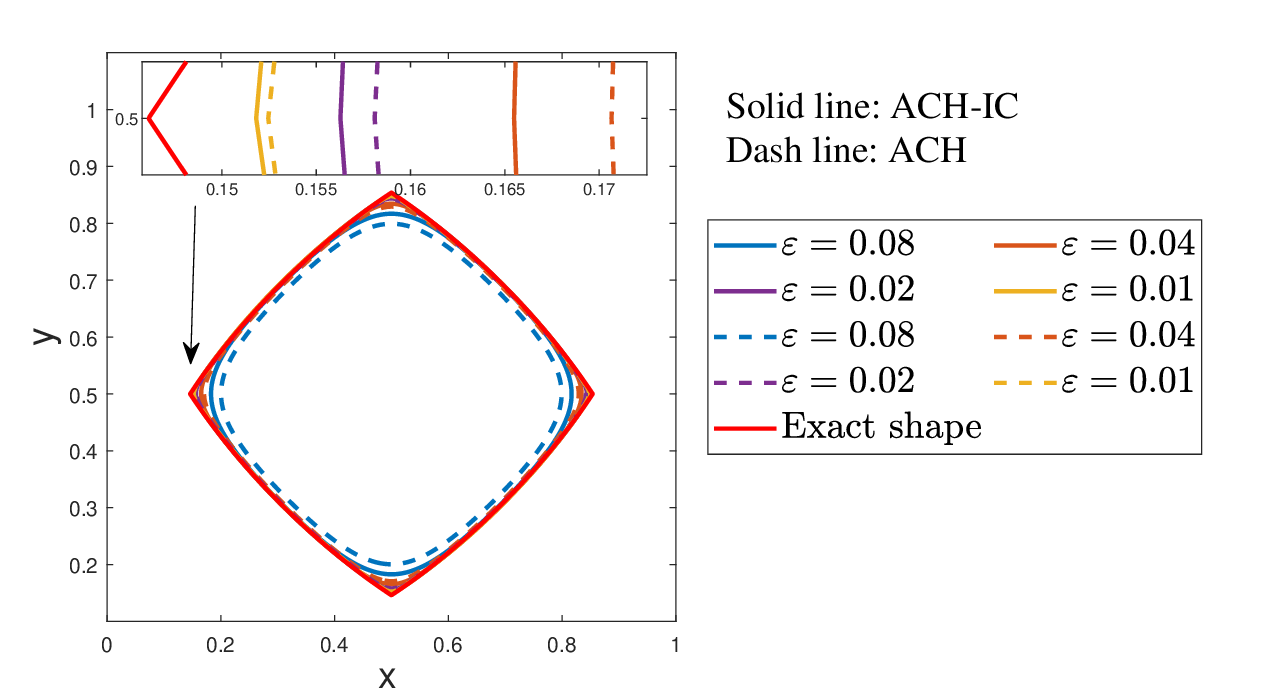}
	\caption{ Equilibrium shapes obtained by the ACH-IC (solid line) and ACH (dashed line) for strongly four-fold anisotropic density ($\alpha=0.2$) with thickness parameters $\varepsilon=0.08,~0.04,~0.02,~0.01$. The zoom-in plot nearby the left corner is shown, illustrating the convergence of both model and highlighting the superior accuracy of the ACH-IC.}
	\label{strongly_equi}
\end{figure}

\subsection{Spontaneous shrinkage of circular drops}
To further evaluate the advantages of the ACH-IC model, we present numerical comparisons with its traditional analogue. The study in Ref.~\cite{2017Spontaneous} identified two major deficiencies of traditional Cahn-Hilliard type models: the ``spontaneous shrinkage'' of an enclosed domain due to its boundary curvature and the existence of a ``critical radius'' below which the enclosed domain eventually vanishes due to diffusion. Hereby, we reproduce the same tests in the isotropic case, i.e., under a four-fold anisotropic density with \(\alpha = 0\). The computational parameters remain consistent with those used in the previous sections, with \(\varepsilon = 0.02\). The computational domain is a square region \(\Omega = [0, 1]^2\). The initial condition is given by a circular drop centered at \((0.5, 0.5)\) with an initial radius \(r_0\):
\[  
u(x, y, t=0) = -\tanh\left(\frac{\sqrt{(x-0.5)^2+(y-0.5)^2} - r_0}{\sqrt{2}\varepsilon}\right).  
\]  
Under this setting, the sharp-interface dynamics corresponding to such an initial circular profile should reach equilibrium at the same initial location with no shrinkage.

As shown in Ref.~\cite{2017Spontaneous}, when \(r_0 \leq \left(\frac{\sqrt{6}\varepsilon|\Omega|}{8\pi}\right)^{1/3} \approx 0.1249\), the circular drop disappears entirely. For \(r_0 > 0.1249\), the drop still shrinks but can be stabilized at an equilibrium state dramatically smaller than the initial circle with significant volume loss. The equilibrium state is reached when \(\|u^{n+1} - u^n\| < 10^{-8}\). To assess performance, we compare the equilibrium shapes computed by the ACH and ACH-IC for initial radii \(r_0 = 0.13\) and \(r_0 = 0.1\) respectively, as shown in Fig.~\ref{fig_circle13} and Fig.~\ref{fig_circle}.  

\begin{figure}[htb]
	\centering
	%	\vspace{-0.4cm}
	%	\setlength{\abovecaptionskip}{0cm} 
	%	\setlength{\belowcaptionskip}{0cm}
	{\includegraphics[trim=0cm 0.0cm 0cm 0cm, clip,width=0.46\linewidth]{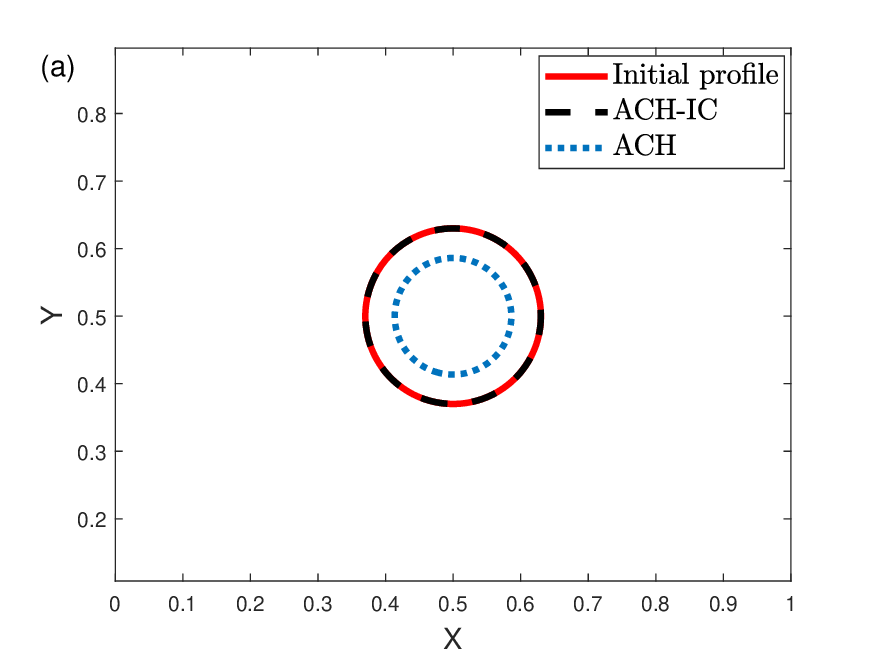}}
	{\includegraphics[trim=0cm 0cm 0cm 0cm, clip,width=0.46\linewidth]{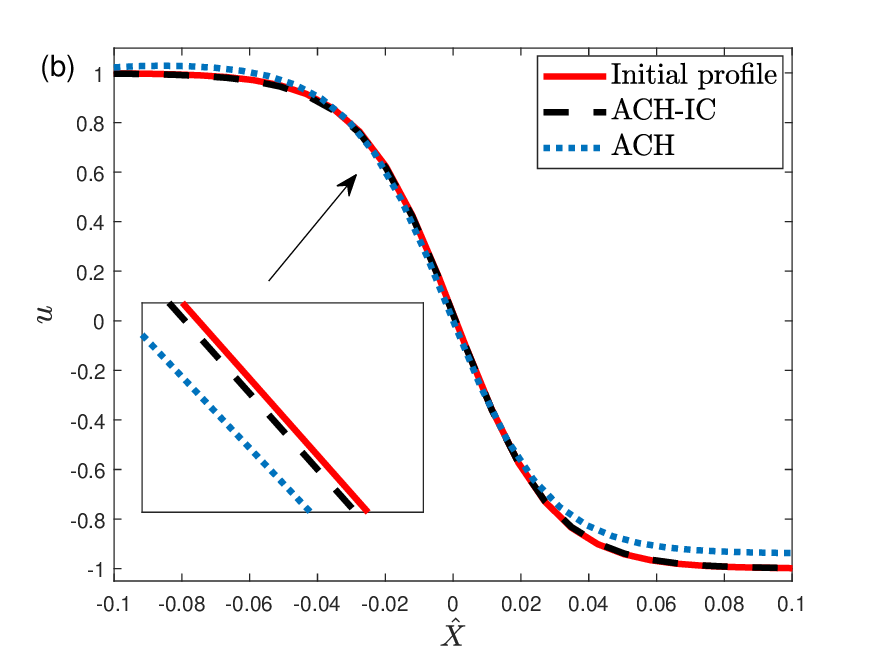}}
	\caption{\label{fig_circle13} (a) Comparisons of equilibrium shapes obtained by the ACH-IC (black line) and the ACH (blue line) in the isotropic case, with the initial circular profile (red solid line) of radii \(r_0 = 0.13\). (b) The cross-sectional profile of $u$ near the right half semi-circle at \( y = 0.5 \), where the shifted horizontal coordinate \( \hat{X} \) represents the distance to the rightmost interface point. The inset plot highlights a detailed comparison near the interface, where the ACH-IC model coincides well with the initial \(\tanh\) profile, whereas the ACH model exhibits a noticeable deviation.}
\end{figure}

\begin{figure}[htb]
	\centering
	%	\vspace{-0.4cm}
	%	\setlength{\abovecaptionskip}{0cm} 
	%	\setlength{\belowcaptionskip}{0cm}
	{\includegraphics[trim=0cm 0.0cm 0cm 0cm, clip,width=0.46\linewidth]{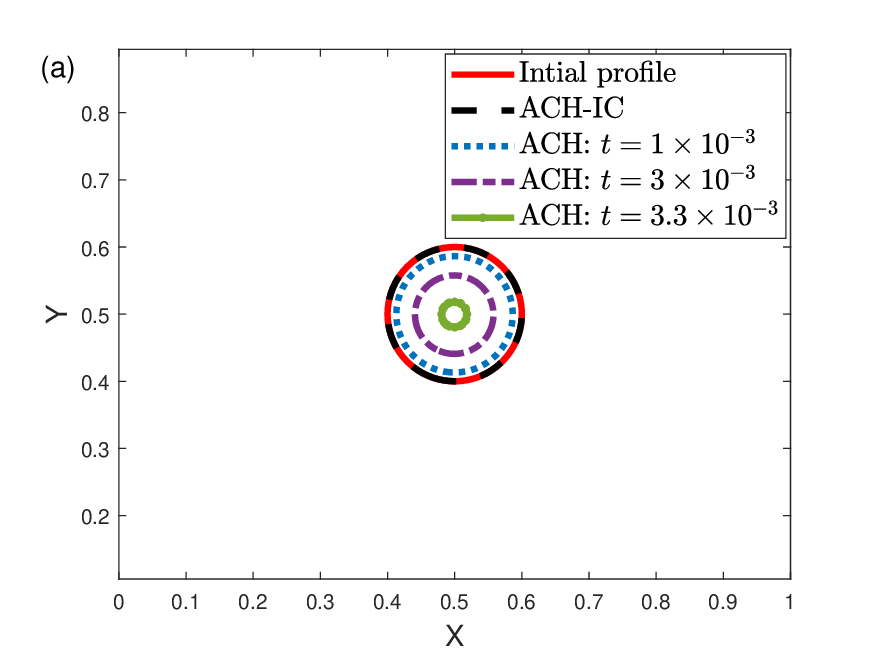}}
	{\includegraphics[trim=0cm 0cm 0cm 0cm, clip,width=0.46\linewidth]{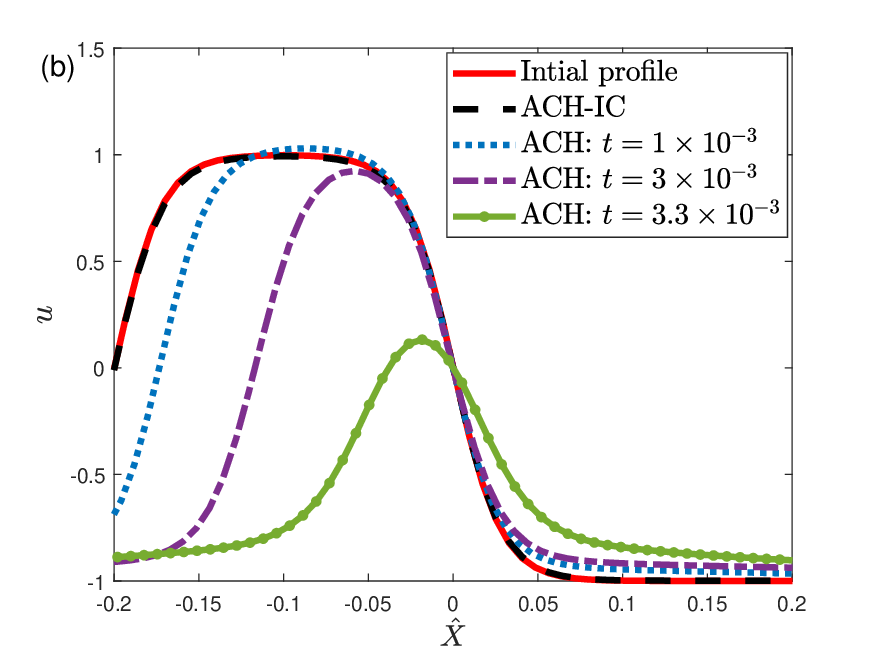}}
	\caption{\label{fig_circle} (a) Comparisons of the interface profiles obtained by the ACH-IC (black line) and the ACH (blue, purple, and green lines) during evolution in the isotropic case, where the initial circular profile (red solid line) has a radius \(r_0 = 0.1\) below the reported critical value. (b) The evolution of \( u \) along the shifted horizontal coordinate \( \hat{X} \), extracted from a cross-section at \( y = 0.5 \). }
\end{figure}

The results in Fig.~\ref{fig_circle13} (a) clearly demonstrate that for \( r_0 = 0.13 \), the circle computed using the ACH exhibits a significant shrinkage, while the ACH-IC preserves the initial profile with high accuracy. In a closer observation, the plot of the cross-sectional profile in Fig.~\ref{fig_circle13} shows perfect preservation of the initial \(\tanh\) profile in the ACH-IC solution during this equilibration process. The significant deviation seen in the ACH solution leads to a pronounced shrinkage that the ACH-IC model does not produce.

For \( r_0 = 0.1 \), a value below the critical radius reported in Ref.~\cite{2017Spontaneous}, the numerical results are presented in Fig.~\ref{fig_circle} (a). The interface of the ACH-IC effectively preserves the initial shape, stabilizing an equilibrium state that closely matches the initial circular profile. In sharp contrast, the interface of the ACH solution exhibits progressive shrinkage over time and eventually collapses at \( t = 3.3 \times 10^{-3} \).  
To further illustrate this computationally undesirable behavior, Fig.~\ref{fig_circle} (b) presents the cross-sectional profiles at \( y = 0.5 \). The transition profile of the ACH solution $u$ drastically deviates from its initial profile during evolution, with far-field values being negative and severely deviated from the ideal value $\pm1$ in bulk phases. These results demonstrate that the evolution of the ACH solution exhibits progressive shrinkage toward its final collapse.
%initially follows the expected profile but gradually deviates, leading to significant distortion at \( t = 3.3 \times 10^{-3} \). 
In contrast, the ACH-IC model remains consistent with the initial \(\tanh\) profile, highlighting its superior ability to preserve the interfacial structure and the enclosed volume.

These numerical experiments highlight the significant improvements achieved by introducing a new conserved quantity in our model. The spontaneous shrinkage effect is effectively mitigated, and the critical radius is notably reduced (the investigation of the critical value in the ACH-IC demands further extensive analytical and numerical studies, which are beyond the scope of this work). This improvement brought by the ACH-IC model allows for the stable simulation of smaller drops in many interfacial phenomena using phase-field models, with the non-physical shrinkage effectively avoided. 
\subsection{Comparisons of the flux distributions}
To further compare the differences between the two models, this subsection focuses on their flux \( \mathbf{J} \) and normal flux \(\mathbf{J}_\nu\) distributions. The tests are conducted under isotropic conditions for simplicity, providing an illustrative and insightful comparison. The computational parameters remain the same as in previous sections with \(\varepsilon = 0.04\). The computational domain is a square region \(\Omega = [0, 1]^2\). The initial condition is  
\[
u(x,y,t=0) = \tanh\left(\frac{-\max( |x-0.5|-0.2,|y-0.5|-0.2)}{\sqrt{2}\varepsilon}  \right),
\]
which represents a square region centered at \((0.5, 0.5)\) with a side length of \(0.4\).  

The normal flux for the ACH is defined by  
\[
\mathbf{J}_\nu := -M(u) \nabla \mu \cdot \frac{\nabla u}{|\nabla u|},
\]  
while the normal flux for the ACH-IC is defined by  
\[
\mathbf{J}_\nu := -M(u) \nabla (N(u) \mu) \cdot \frac{\nabla u}{|\nabla u|}.
\]

\begin{figure}[htb]
	\centering
	\includegraphics[trim=0.5cm 0.6cm 0.5cm 0.5cm, clip,width=1\linewidth]{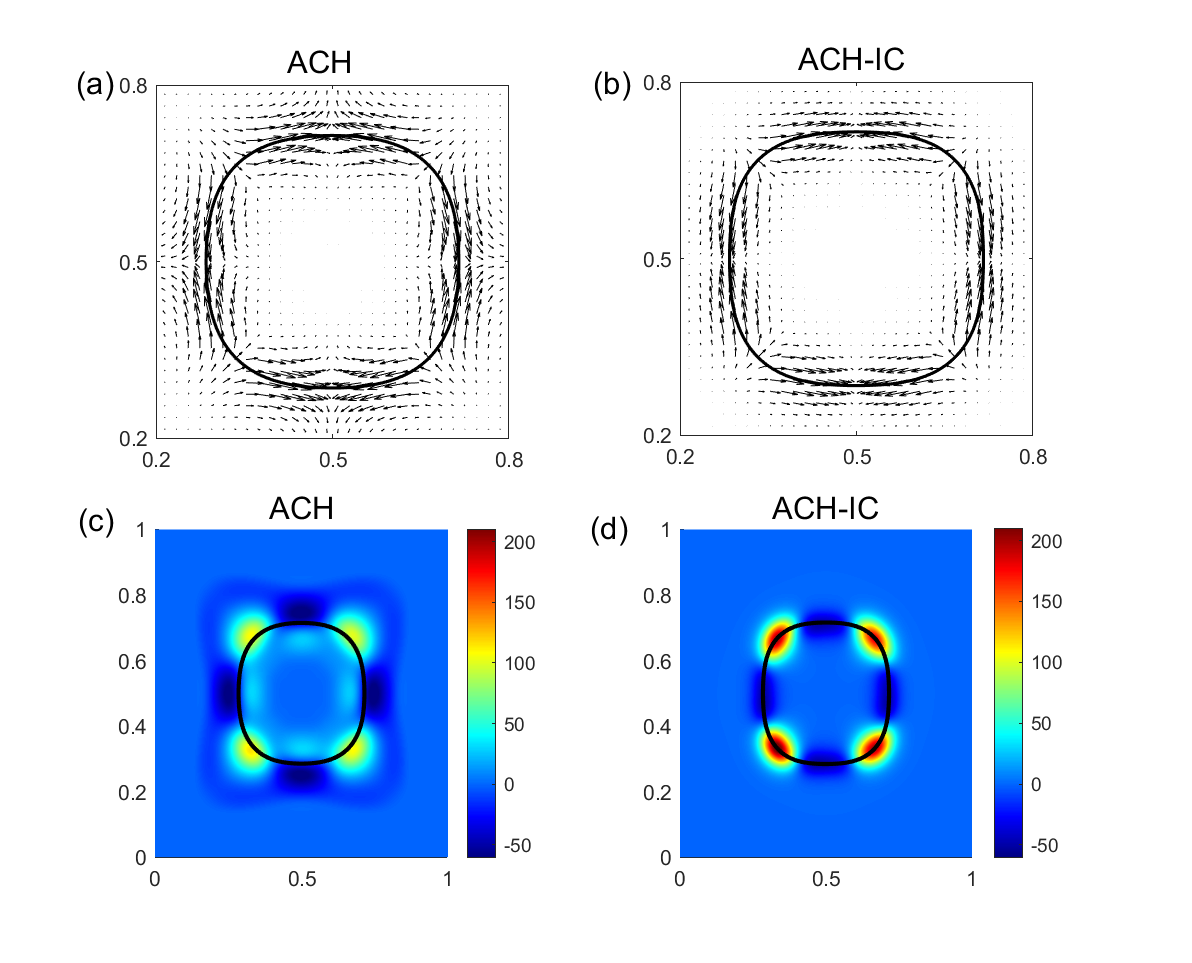}
	\caption{\label{fig_flux_small} (a) and (b) show the flux $\mathbf{J}$ distribution for the ACH and ACH-IC models at an early stage of evolution (\( t = 2 \times 10^{-5} \)) in the isotropic case. The arrows indicate the direction and magnitude of the flux, while the black solid line represents the interface. (c) and (d) show the corresponding normal flux $\mathbf{J}_\nu$ distribution.}
\end{figure}

\begin{figure}[htb]
	\centering
	\includegraphics[trim=0.5cm 0.6cm 0.5cm 0.5cm, clip,width=1\linewidth]{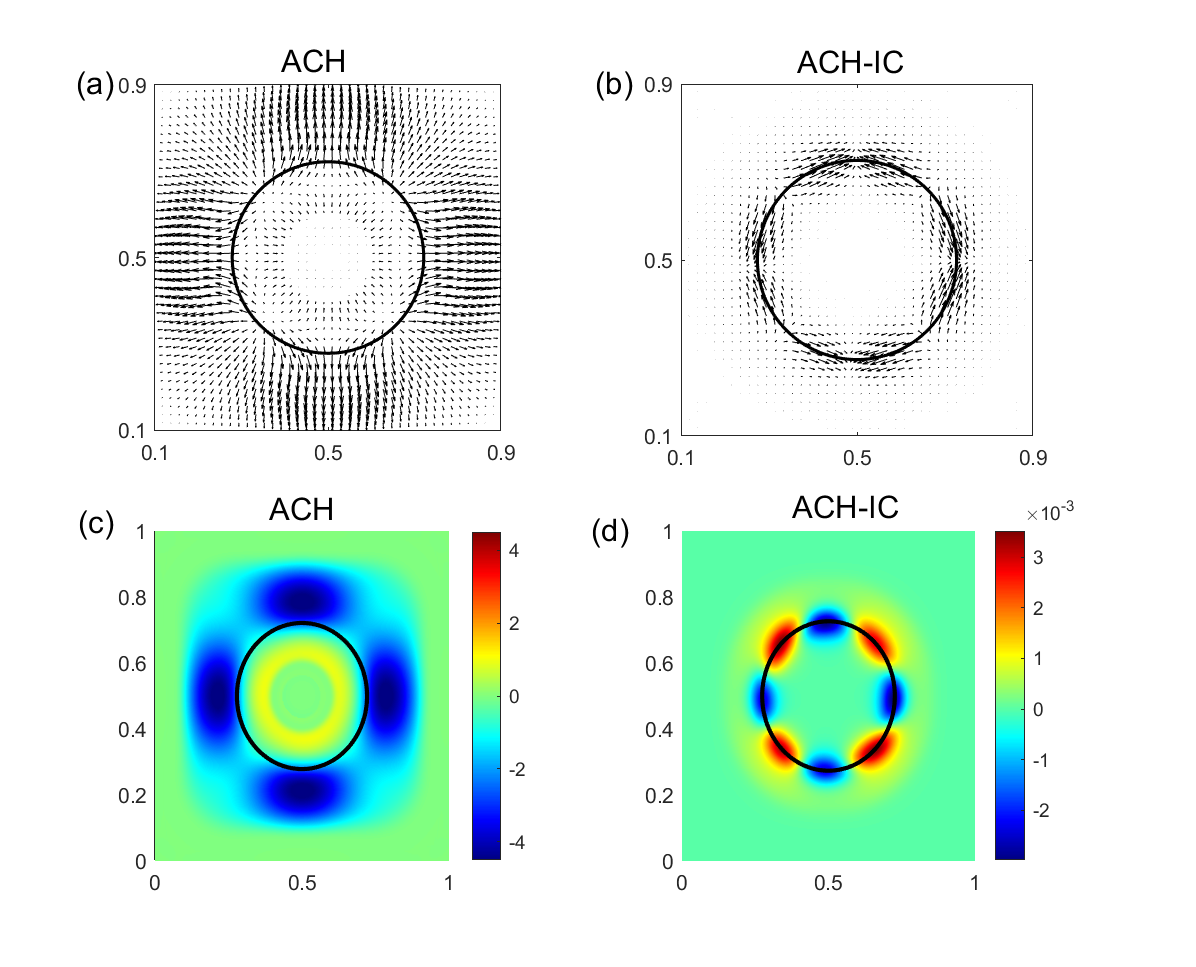}
	\caption{\label{fig_flux_large}(a) and (b) show the flux $\mathbf{J}$ distribution for the ACH and ACH-IC models at a later stage of evolution (\( t = 2 \times 10^{-4} \)) in the isotropic case. The arrows indicate the direction and magnitude of the flux, while the black solid line represents the interface. (c) and (d) show the corresponding normal flux $\mathbf{J}_\nu$ distribution.}
\end{figure}

Fig.~\ref{fig_flux_small} (a) and (b) show the flux distributions for the ACH and ACH-IC models in the early stage of evolution (\( t = 2 \times 10^{-5} \)). The arrows indicate the direction and magnitude of the flux, while the black solid line denotes the interface. 
For both models, the flux field demonstrates pronounced divergence at the square's corners and appreciable convergence towards the midpoints of the edges. This flux distribution effectively drives the interface evolution, characterized by the phase variable migration from the corner regions toward the edge centers, leading to a progressively rounded interface.
Fig.~\ref{fig_flux_small} (a) and (b) are visibly identical, with only a tiny difference observed for the normal flux distribution near the midpoints of the edges, i.e., the ACH-IC produces smaller normal flux there. 

Fig.~\ref{fig_flux_small} (c) and (d) present the normal flux distributions. We see that the ACH-IC model produces the normal flux which is concentrated around the interface, while the ACH model produces the normal flux which is distributed over a wider region. This comparison shows that the ACH-IC results are more aligned with the characteristics of surface diffusion.

%the normal flux distribution follows a characteristic pattern where the corners of the square exhibit positive values ($\mathbf{J}_\nu>0$) indicating phase variable outflow that leads to local shrinkage of the interface. Meanwhile, the middle sections of the edges display negative values ($\mathbf{J}_\nu<0$), signifying phase variable influx, which promotes local expansion of the interface. 
% This alternate flux distribution transforms the interface from an initial rectangular shape to a more rounded configuration. The difference is that for the ACH-IC, the normal flux is almost concentrated near the interface, while for the ACH, the normal flux is distributed over a wider range. Therefore, the ACH-IC is more consistent with the characteristics of surface diffusion.

Fig.~\ref{fig_flux_large} shows that in the late stage of the evolution ($t=2\times10^{-4}$), the flux and normal flux distributions obtained from the ACH and ACH-IC models exhibit distinct characteristics, which give rise to different interface evolution behaviors and contrasting volume conservation properties.
Compared to the early stage shown in Fig.~\ref{fig_flux_small}, now the overall flux magnitude in both models has decreased by several orders of magnitude, indicating that the surface diffusion has slowed down significantly in approaching the equilibrium. 
Besides this magnitude reduction, there are noticeable differences between the two models. In the ACH, a substantial amount of outward flux is observed outside the interface, showing a net outflow of phase variables and hence a further interface contraction. Therefore, the ACH struggles with volume conservation which suffers a gradual shrinkage of the enclosed region.
In contrast, the ACH-IC maintains a more localized and balanced flux distribution, with alternating positive and negative normal fluxes narrowly concentrated along the interface. This improved flux regulation helps the ACH-IC achieve better volume preservation in comparison to the ACH. 
The contrast in flux distribution between the two models demonstrates the importance and impact of the flux stabilization mechanism in governing long-term interface evolution, with the advantage of the ACH-IC approach manifested in evolving a more accurate interfacial structure. This will be further evidenced by the results in the next subsection.

\subsection{Evolution of a ``flower'' shape}
Regarding the numerical evolution, we compare the ACH solution and the ACH-IC solution with the numerical solution to the sharp-interface model as reference, subject to a complex initial condition, i.e., a ``flower'' shape, as depicted in Fig.~\ref{flower}. Specifically, the initial condition is given by 
\[
u(x,y,t=0)=-\tanh\left(\frac{\sqrt{x^2+y^2}-( 1-0.65\sin(7\,\text{atan2}(y,x)))}{\sqrt{2}\varepsilon}\right).
\]
\begin{figure}[htb]
	\centering
	\vspace{-0.4cm}
	\begin{minipage}{0.24\linewidth}
		{\includegraphics[trim=2.3cm 0cm 2.5cm 0.0cm, clip,width=1\linewidth]{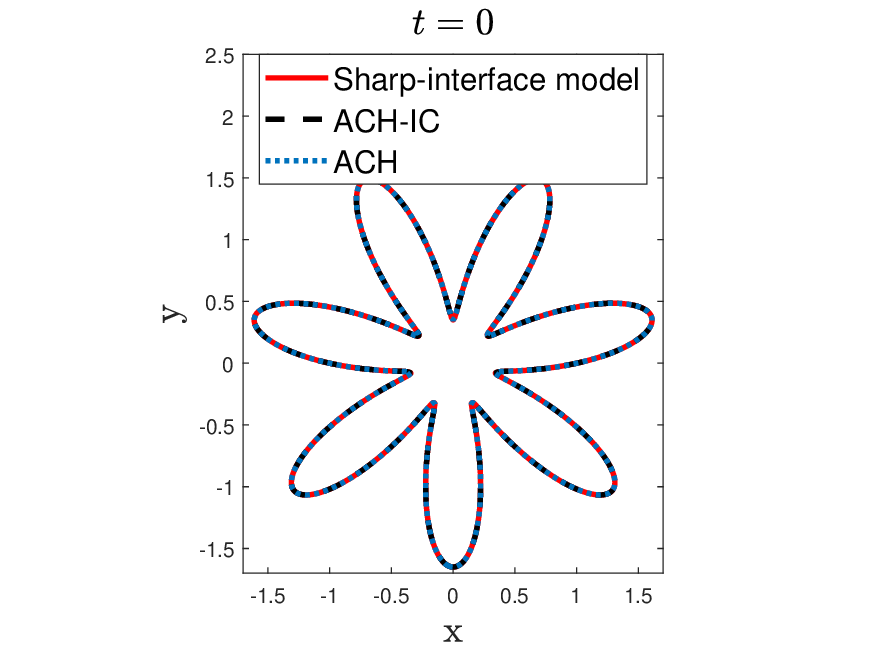}}
%		\caption*{(a) Initial condition.}
	\end{minipage}
	\begin{minipage}{0.24\linewidth}
		{\includegraphics[trim=2.3cm 0cm 2.5cm 0.0cm, clip,width=1\linewidth]{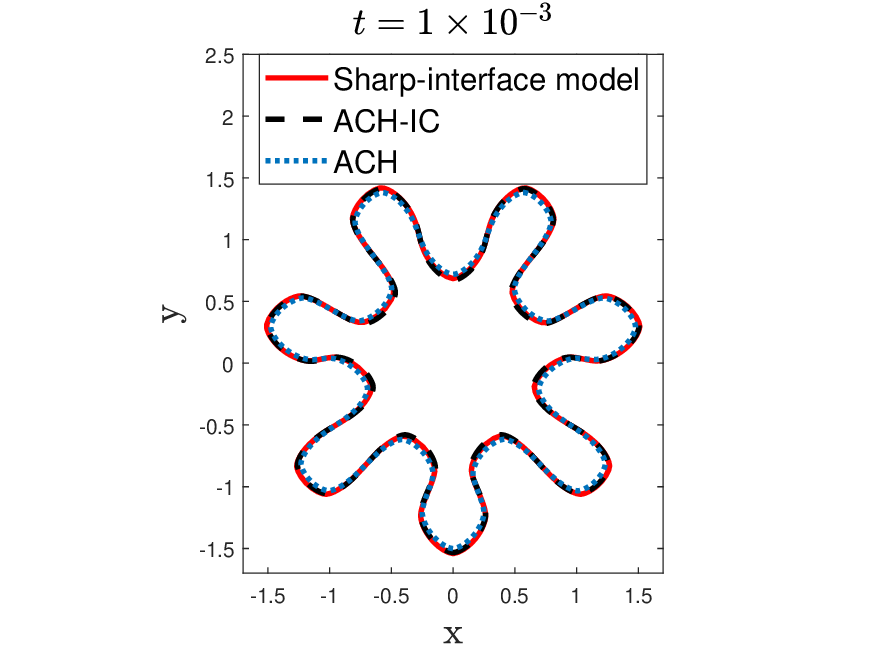}}
%		\caption*{(b) $ t=1\times10^{-3} $.}
	\end{minipage}
	\begin{minipage}{0.24\linewidth}
		{\includegraphics[trim=2.3cm 0cm 2.5cm 0.0cm, clip,width=1\linewidth]{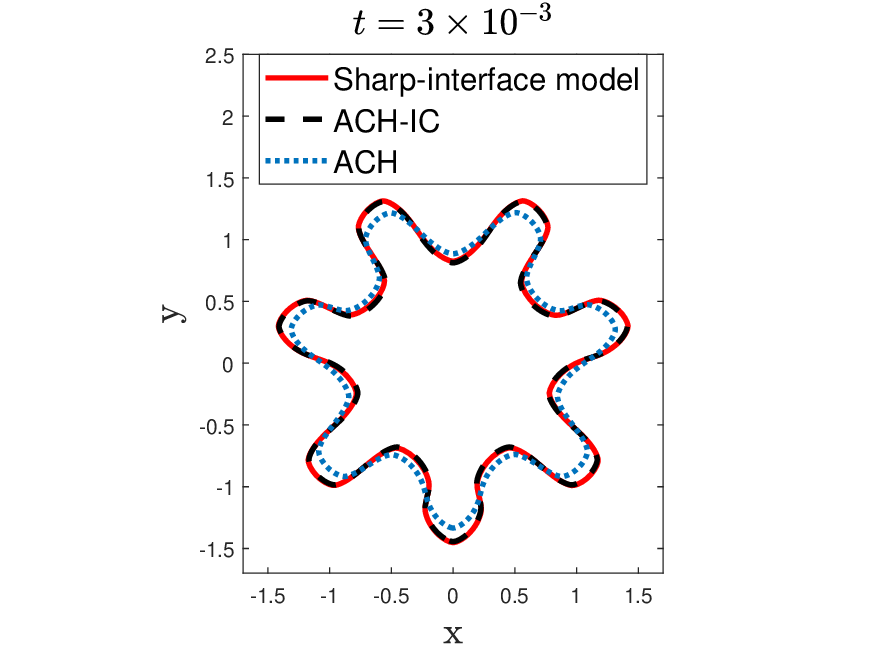}}
%		\caption*{(c) $ t=3\times10^{-3} $.}
	\end{minipage}
	\begin{minipage}{0.24\linewidth}
		{\includegraphics[trim=2.3cm 0cm 2.5cm 0.0cm, clip,width=1\linewidth]{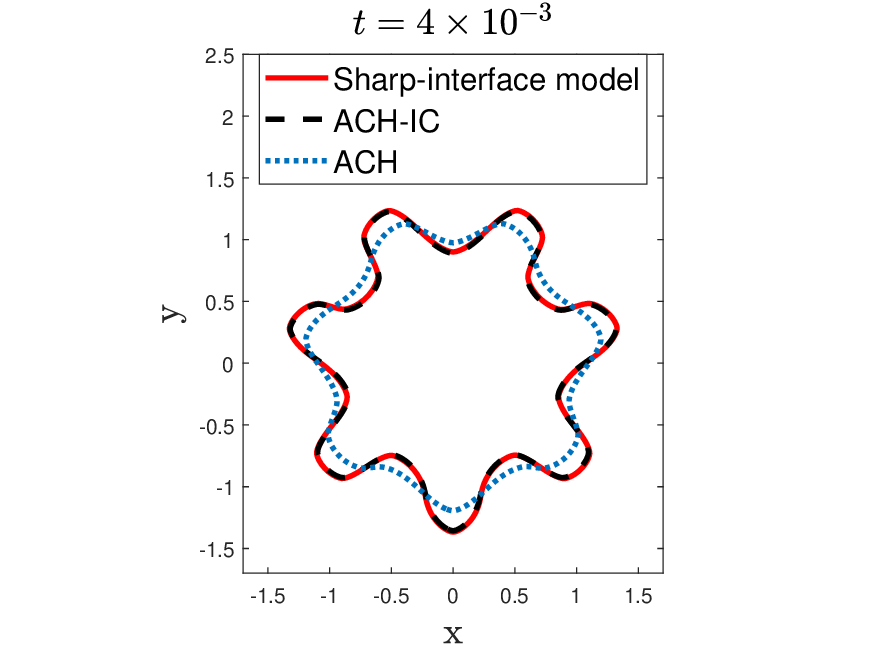}}
%		\caption*{(d) $ t=4\times10^{-3} $.}
	\end{minipage}\\
	\begin{minipage}{0.24\linewidth}
		{\includegraphics[trim=2.3cm 0cm 2.5cm 0.0cm, clip,width=1\linewidth]{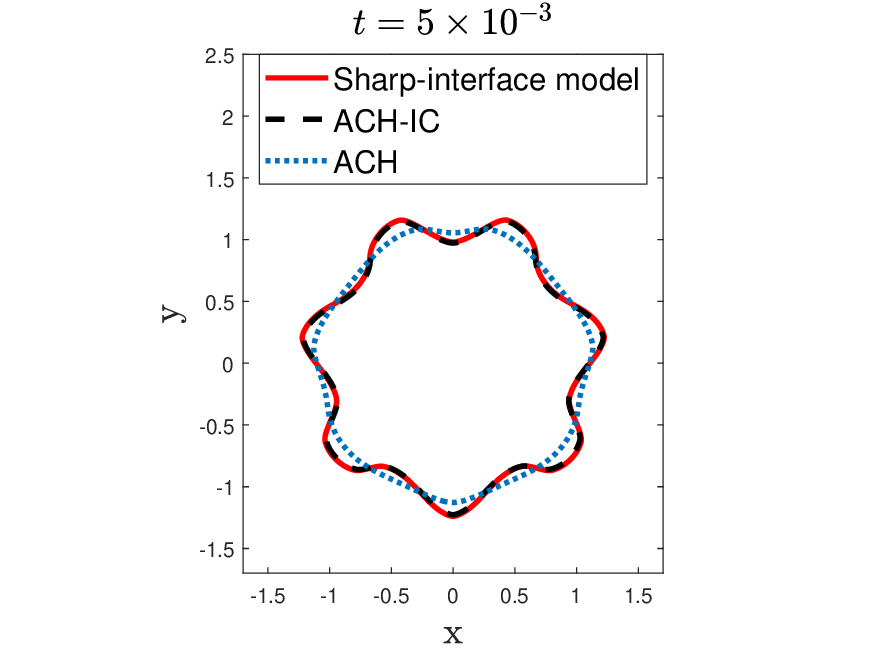}}
%		\caption*{(e) $ t=5\times10^{-3} $.}
	\end{minipage}
	\begin{minipage}{0.24\linewidth}
		{\includegraphics[trim=2.3cm 0cm 2.5cm 0.0cm, clip,width=1\linewidth]{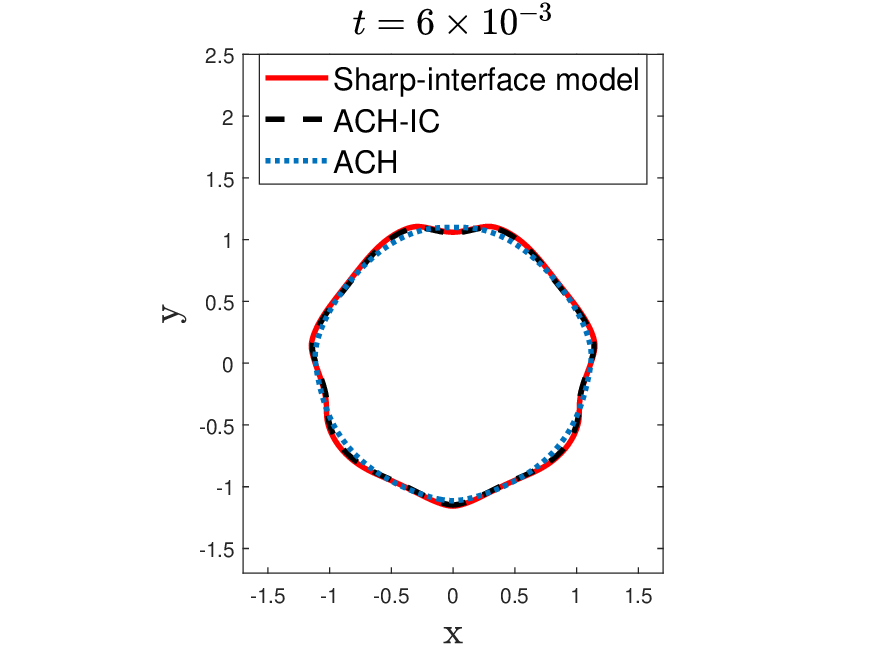}}
%		\caption*{(f) $ t=6\times10^{-3} $.}
	\end{minipage}
	\begin{minipage}{0.24\linewidth}
		{\includegraphics[trim=2.3cm 0cm 2.5cm 0.0cm, clip,width=1\linewidth]{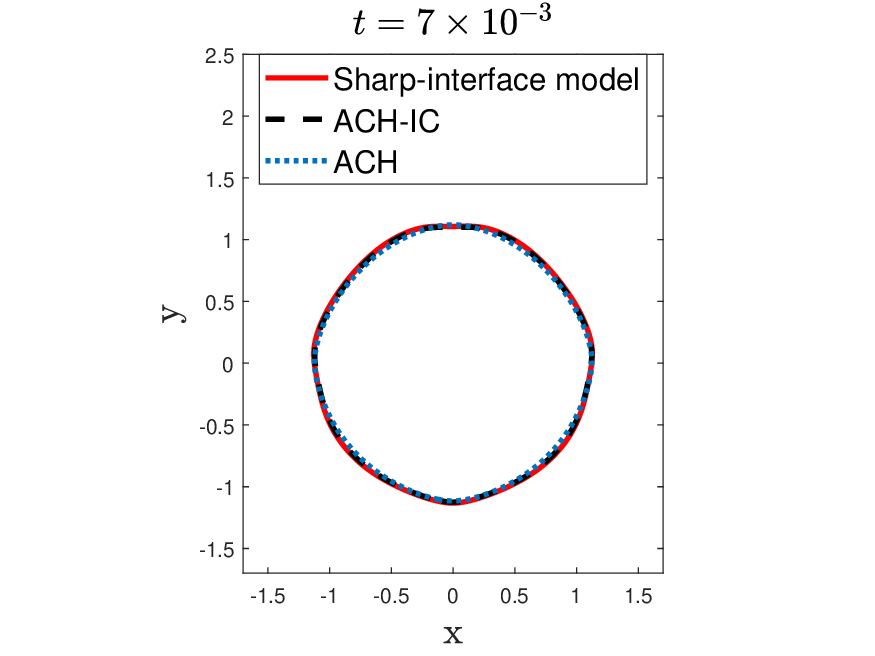}}
%		\caption*{(g) $ t=7\times10^{-3} $.}
	\end{minipage}
	\begin{minipage}{0.24\linewidth}
		{\includegraphics[trim=2.3cm 0cm 2.5cm 0.0cm, clip,width=1\linewidth]{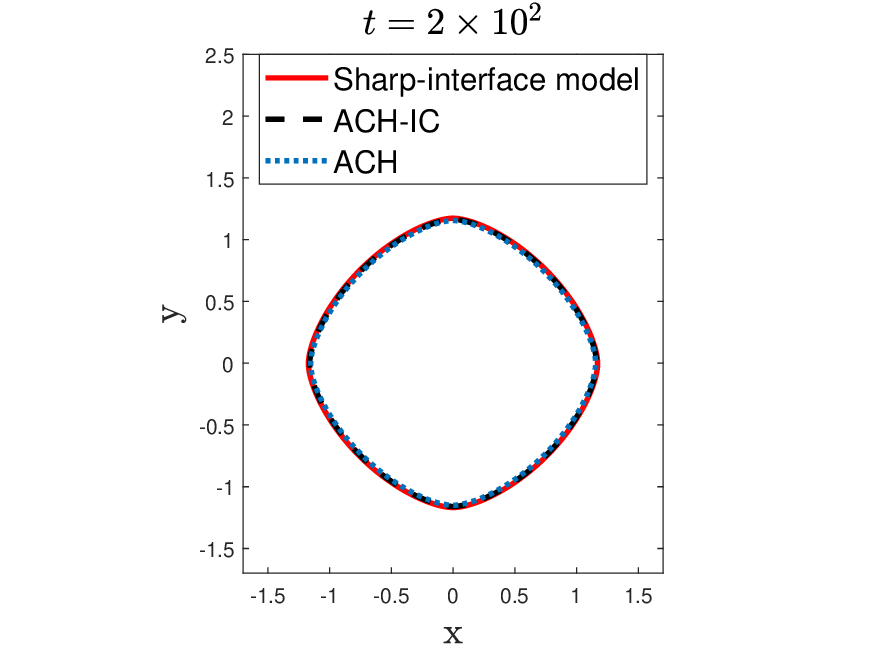}}
%		\caption*{(h) $ t=2\times10^{-2} $.}
	\end{minipage}	
	%	\vspace{-0.4cm}
	\caption{Comparisons about the evolution of an initially ``flower'' shape toward its equilibrium at different time levels among the sharp-interface model, the ACH-IC, and the ACH.}
	\label{flower}
\end{figure}

We take the four-fold anisotropic density \eqref{four_fold} with $\alpha=0.05$. For the phase-field model, the computational domain is a square region $\Omega=[-1.72,1.72]^2$. During the simulation, we set $\varepsilon=0.06$, the mesh grid $256\times256$, and the time step $\delta t=10^{-6}$. We multiply the ACH  by a time scale $4/9$\cite{dziwnik2017anisotropic} such that its sharp-interface velocity is the same as the velocity of the ACH-IC and the sharp-interface model. We adopt the parametric finite element method for the sharp-interface model to simulate its numerical evolution (for details about the algorithm, see the reference \cite{jiang17jcp}).  We can observe that the ACH-IC solution can match the sharp-interface solution very well during the whole evolution process, while the ACH solution would have larger deviations from the reference at the early stage of evolution, especially at the segments where the curvature changes signs.

\begin{figure}[htb]
	\centering
	\includegraphics[trim=0.5cm 0.0cm 0.5cm 0.5cm, clip,width=0.5\linewidth]{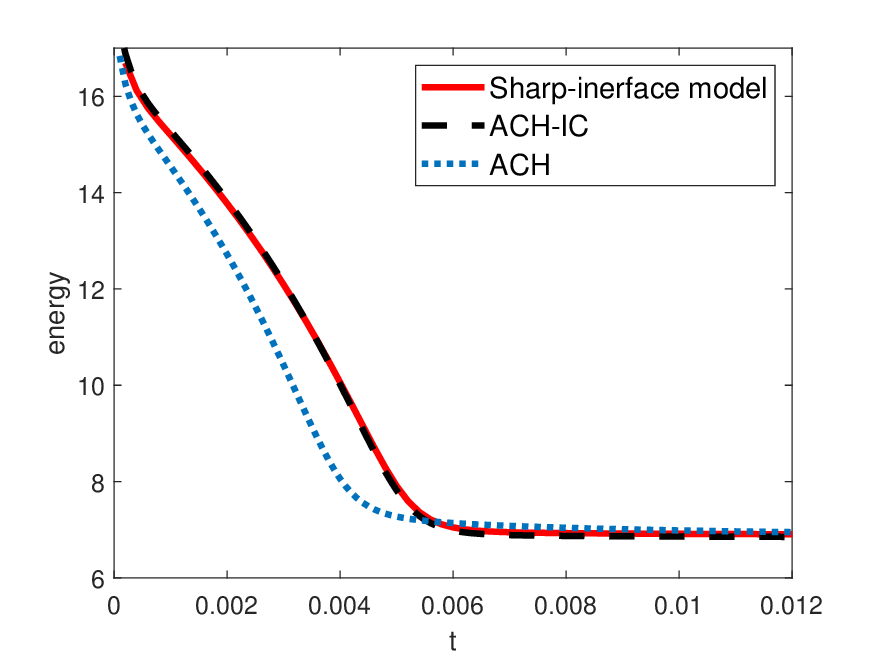}
	\caption{\label{energy_flower} comparison in energy evolution  of an initially ``flower'' shape toward its equilibrium for the sharp-interface model (in red), ACH-IC (in black), and ACH (in blue). }
\end{figure}

In addition, Fig.~\ref{energy_flower} presents the numerical energy decay for the three models: the sharp-interface model (in red), ACH-IC (in black), and ACH (in blue). 
The energy in the sharp-interface model is given by Equation \eqref{original_aniso_energy}, while the energy in the ACH and ACH-IC is given by Equation \eqref{Torabi_energy} multiplied by a mixing energy density $\lambda_m=\frac{3\sqrt{2}}{4}$\cite{modica1977limite,lussardi2015note}.
Initially, the energy decreases rapidly for all three models, indicating the system's evolution. As time proceeds, the energy evolution curves for the ACH-IC and the sharp-interface model agree remarkably, both showing a consistent and smooth decrease. In contrast, the energy evolution of the ACH deviates appreciably from that of the other two models. Overall, compared to the other two nearly identical curves, the ACH curve shows a faster approach toward a steady value. This quantitatively appreciable difference in energy evolution suggests that the ACH-IC and the sharp-interface model can produce energy dissipation rates with remarkable agreement, implying that the ACH-IC can describe surface diffusion in a more accurate way beyond the reach of ACH.

\subsection{Pinch-off dynamics of a thin tube}   	
Finally, we investigate the pinch-off dynamics of a thin tube in a computational domain $[0,1]\times[0,\frac{1}{4}]\times[0,\frac{1}{4}]$ to show the importance of possessing second-order real volume conservation. We consider the four-fold anisotropic density \eqref{four_fold} with $\alpha=0.05$. To perform the simulations, we set $\varepsilon=0.02$, the spatial mesh size $\delta x=1/256$, and the time step $\delta t=10^{-8}$. The initial shape is a tube that is very long and thin, defined by
\[
u(x,y,z,t=0)=\tanh\left(\frac{-\max( |x-1/2|-0.4,|y-1/8|-0.02,|z-1/8|-0.02)}{\sqrt{2}\varepsilon}  \right).
\]

\begin{figure}[!htbp]
	\centering
	%	\vspace{-0.4cm}
	%	\setlength{\abovecaptionskip}{0cm} 
	%	\setlength{\belowcaptionskip}{0cm}
	%	{\includegraphics[trim=0cm 0.8cm 0cm 1.5cm, clip,width=0.46\linewidth]{figures/NMN_tube_0}}
	%	{\includegraphics[trim=0cm 0.8cm 0cm 1.5cm, clip,width=0.46\linewidth]{figures/NMN_tube_0}}\\
	{\includegraphics[trim=0cm 0.8cm 0cm 1.5cm, clip,width=0.46\linewidth]{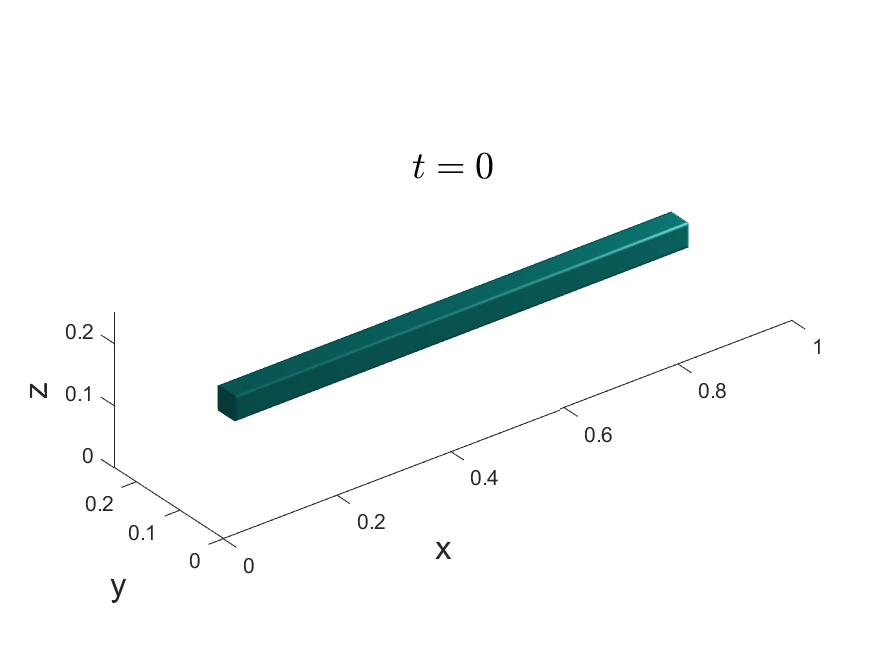}}
	{\includegraphics[trim=0cm 0.8cm 0cm 1.5cm, clip,width=0.46\linewidth]{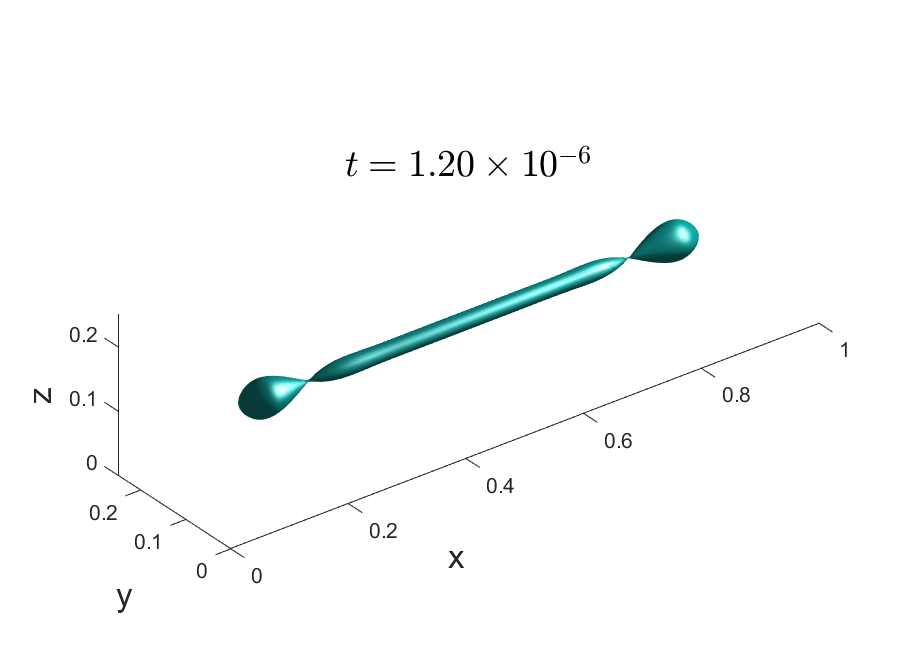}}\\
	{\includegraphics[trim=0cm 0.8cm 0cm 1.5cm, clip,width=0.46\linewidth]{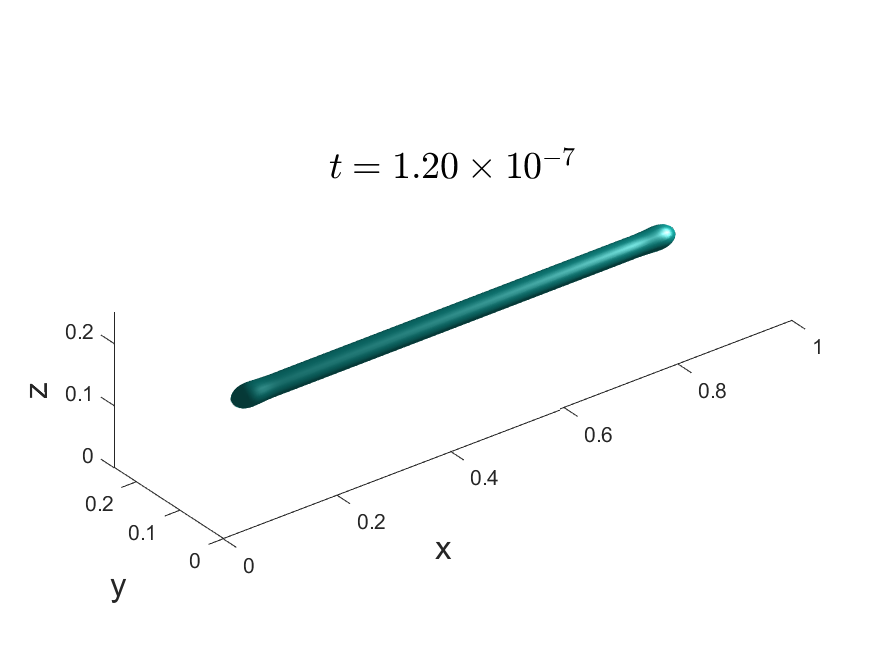}}
	{\includegraphics[trim=0cm 0.8cm 0cm 1.5cm, clip,width=0.46\linewidth]{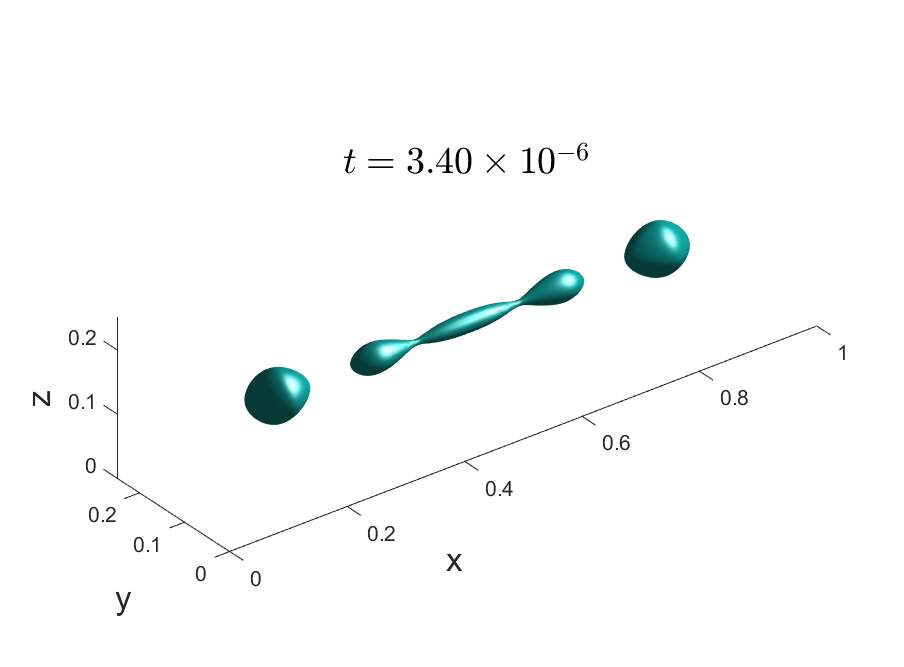}}\\
	{\includegraphics[trim=0cm 0.8cm 0cm 1.5cm, clip,width=0.46\linewidth]{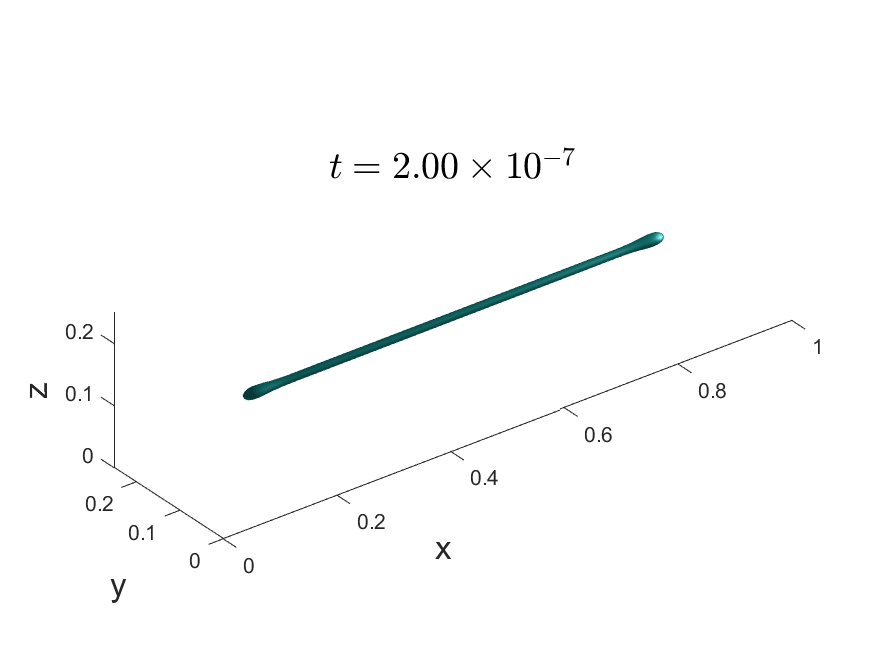}}
	{\includegraphics[trim=0cm 0.8cm 0cm 1.5cm, clip,width=0.46\linewidth]{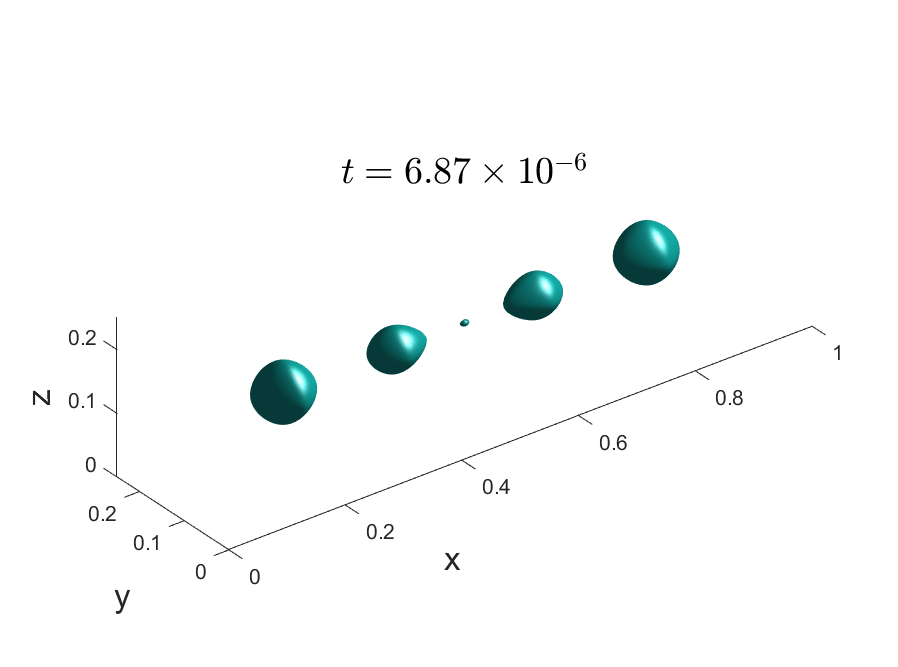}}
	\\
	{\includegraphics[trim=0cm 0.8cm 0cm 1.5cm, clip,width=0.46\linewidth]{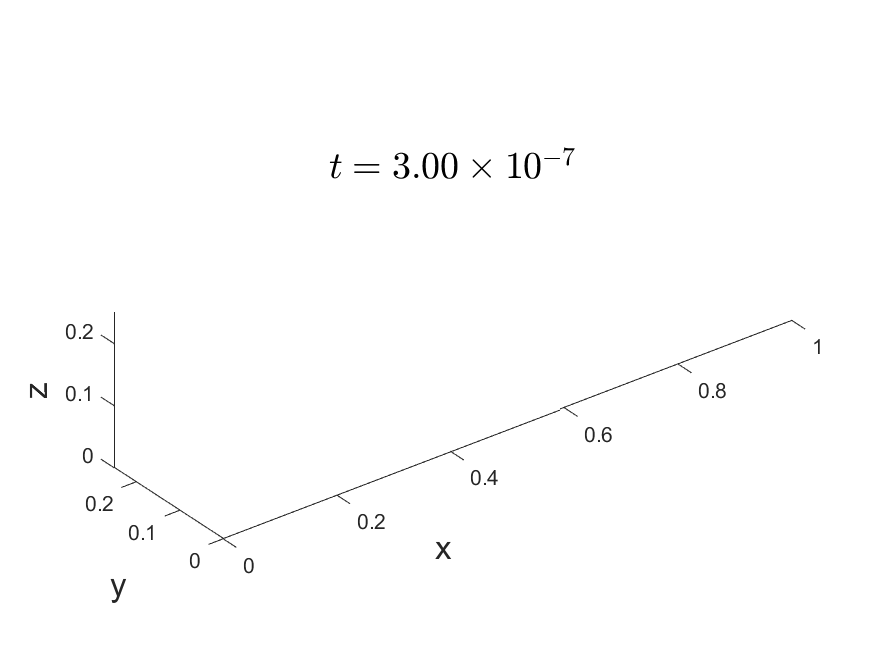}}
	{\includegraphics[trim=0cm 0.8cm 0cm 1.5cm, clip,width=0.46\linewidth]{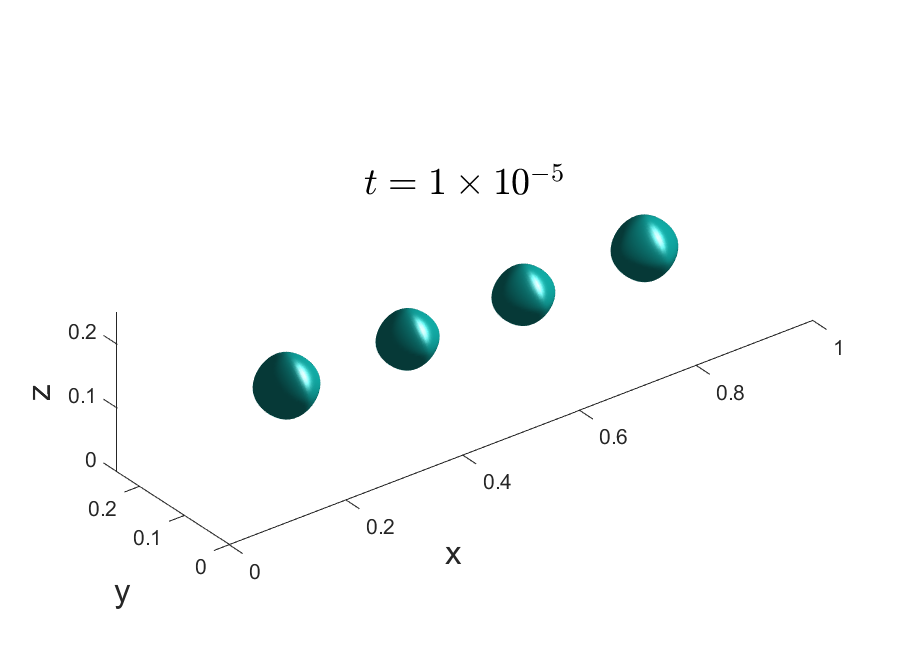}}
	\caption{Several snapshots depicting the simulation of the pinch-off dynamics of a thin tube using the ACH (left column) and the ACH-IC (right column).}
	\label{thinthinfilm}
\end{figure}

Fig.~\ref{thinthinfilm} illustrates the evolution processes obtained from the ACH model (shown on the left) and the ACH-IC model (shown on the right).
For the ACH, the surface disappears completely in the end, contradicting the behavior governed by surface diffusion. This discrepancy arises probably due to the same reason as in the spontaneous shrinkage of a small drop. Furthermore, the value of $\varepsilon$ used here is not small enough for the ACH to resolve the pinch-off process.
For the ACH-IC, we can clearly observe the pinch-off phenomenon and the generation of small drops that proceed through surface diffusion\cite{jiang2012phase,backofen2019convexity}.
Although similar numerical results can be obtained by employing the ACH with a much smaller $\varepsilon$, a much denser mesh is usually required to resolve the interface and hence higher computational cost is inevitable. These numerical results also show that the thickness of the thin tube must be greater than $2\varepsilon$ when the ACH model is employed to simulate pinch-off dynamics. 
Despite the excellent performance in the simulation of pinch-off dynamics by the ACH-IC model, an interesting phenomenon should be noted that the small drop in the middle is absorbed by the four big drops and eventually disappears, a phenomenon known as the Ostwald ripening \cite{Voorhee1985Ostwald}. This defect in simulating dynamics with very small drops can be relieved if a smaller $\varepsilon$ is adopted, or higher-order volume-preserving models could be exploited in accurately capturing surface diffusion processes.

%\begin{figure}[H]
%	\centering
%	\vspace{-0.4cm}
%	\setlength{\abovecaptionskip}{0cm} 
%	\setlength{\belowcaptionskip}{0cm}
%	{\includegraphics[trim=0cm 3cm 0cm 3cm, clip,width=0.46\linewidth]{figures/initial.eps}}
%	{\includegraphics[trim=0cm 3cm 0cm 3cm, clip,width=0.46\linewidth]{figures/initial.eps}}\\
%	{\includegraphics[trim=0cm 3cm 0cm 3cm, clip,width=0.46\linewidth]{figures/91e_7.eps}}
%	{\includegraphics[trim=0cm 3cm 0cm 3cm, clip,width=0.46\linewidth]{figures/156e_6.eps}}\\
%	{\includegraphics[trim=0cm 3cm 0cm 3cm, clip,width=0.46\linewidth]{figures/1e_6.eps}}
%	{\includegraphics[trim=0cm 3cm 0cm 3cm, clip,width=0.46\linewidth]{figures/348e_6.eps}}\\
%	{\includegraphics[trim=0cm 3cm 0cm 3cm, clip,width=0.46\linewidth]{figures/11e_6.eps}}
%	{\includegraphics[trim=0cm 3cm 0cm 3cm, clip,width=0.46\linewidth]{figures/5e_6.eps}}
%	\\
%	{\includegraphics[trim=0cm 3cm 0cm 3cm, clip,width=0.46\linewidth]{figures/26e_6.eps}}
%	{\includegraphics[trim=0cm 3cm 0cm 3cm, clip,width=0.46\linewidth]{figures/674e_6.eps}}\\
%	{\includegraphics[trim=0cm 3cm 0cm 3cm, clip,width=0.46\linewidth]{figures/318e_6.eps}}
%	{\includegraphics[trim=0cm 3cm 0cm 3cm, clip,width=0.46\linewidth]{figures/1e_5.eps}}
%	\caption{Several snapshots depicting the simulation of a long rectangle using the M-ACH model (left column) and the NMN-ACH model (right column).}
%	\label{thinthinfilm}
%\end{figure}
\section{Conclusion and discussion}\label{chapter5}
In this work, we present a novel phase-field model that incorporates a new conserved quantity to overcome the limitations of real volume loss and spontaneous shrinkage in traditional Cahn-Hilliard-type models. By employing both variational analysis and matched asymptotic analysis, we establish that our model satisfies second-order conservation in real volume, a property essential for accurately capturing the dynamics of small drops. 

Extensive numerical results in both two and three dimensions are presented to show the superior performance of the newly proposed ACH-IC model in comparison with its traditional analogue. In particular, the ACH-IC model significantly mitigates the spontaneous shrinkage of small drops and drastically reduces the critical radius required for stable simulations. Notably, the ACH-IC model effectively preserves small drops during pinch-off dynamics and provides efficient and accurate simulations for complex surface diffusion with larger parameter $\varepsilon$.

\appendix	

\section{Inner expansion of the chemical potential}\label{expansion_second}
%Since in inner region, $\partial_{\rho} U\neq0$, no matter to assume $ \partial_{\rho} U>0 $. 
Before we expand the right hand of \eqref{eq2}, we need to derive some useful equalities. 
By \eqref{coordinates_change}, we have
\begin{equation}\label{expand_n}
	\mathbf{n}=\frac{\nabla u}{|\nabla u|}=\mathbf{\nu}+\varepsilon\frac{\nabla_\mathbf{s}U_0}{\partial_\rho U_0}+\mathcal{O}(\varepsilon^2)
\end{equation}
Using \eqref{equ_xi}, we can expand $ \hat{\gamma}(\mathbf{n}) $ as
\begin{equation}\label{expand_hatgamma}
	\begin{aligned}
		\hat{\gamma}(\mathbf{n})&=\hat{\gamma}\left(\mathbf{\nu}+\varepsilon\frac{\nabla_\mathbf{s}U_0}{\partial_\rho U_0}+\mathcal{O}(\varepsilon^2)\right)\\
		&=\hat{\gamma}(\mathbf{\nu})+\varepsilon{\bm\xi}(\mathbf{\nu})\cdot\frac{\nabla_\mathbf{s}U_0}{\partial_\rho U_0}+\mathcal{O}(\varepsilon^2)\\
		&=\gamma_0+\varepsilon\frac{\gamma_1}{\partial_{\rho} U_0}+\mathcal{O}(\varepsilon^2),
	\end{aligned}
\end{equation}
where we have introduced the short notations $ \gamma_0:=\hat{\gamma}(\mathbf{\nu}) $ and $ \gamma_1:={\bm\xi}(\mathbf{\nu})\cdot\nabla_\mathbf{s}U_0 $ for brevity. Similarly, we can expand $ {\bm\xi}(\mathbf{n}) $ as
\begin{equation}\label{expand_xi}
	\begin{aligned}
		{\bm\xi}(\mathbf{n})=&{\bm\xi}\left(\mathbf{\nu}+\varepsilon\frac{\nabla_\mathbf{s}U_0}{\partial_\rho U_0}+\mathcal{O}(\varepsilon^2)\right)
		={\bm\xi}(\mathbf{\nu})+\frac{\varepsilon}{\partial_{\rho}U_0}\nabla{\bm\xi}(\mathbf{\nu})\nabla_\mathbf{s}U_0+\mathcal{O}(\varepsilon^2).
	\end{aligned}
\end{equation}

For the first term of \eqref{eq2} on the right hand side, we have
\begin{equation}\label{right_mu_first}
	\hat{\gamma}(\mathbf{n})F'(u)=\gamma_0F'(U_0)+\varepsilon\left(\gamma_0F''(U_0)U_1+\gamma_1\frac{F'(U_0)}{\partial_{\rho}U_0}\right)+\mathcal{O}(\varepsilon^2).
\end{equation}
Combining \eqref{expand_n}, \eqref{expand_hatgamma} and \eqref{expand_xi}, we can expand $ \mathbf{m} $ as 
{\small
	\[
	\begin{aligned}	&\mathbf{m}=\left(\varepsilon^{-1}\gamma_0\partial_{\rho}U_0\mathbf{\nu}+\left(\gamma_0\partial_{\rho}U_1\mathbf{\nu}+\gamma_0\nabla_\mathbf{s}U_0+\gamma_1\mathbf{\nu}\right)+\mathcal{O}(\varepsilon)\right)+\\
		&\bigg(\varepsilon^{-1}\partial_{\rho}U_0\left({\bm\xi}(\mathbf{\nu})-\gamma_0\mathbf{\nu}\right)+\left(\nabla{\bm\xi}(\mathbf{\nu})\nabla_\mathbf{s}U_0-\gamma_1\mathbf{\nu}-\gamma_0\nabla_\mathbf{s}U_0	\right)+\mathcal{O}(\varepsilon)\bigg)	\left(\frac{1}{2}+\frac{F(U_0)}{(\partial_{\rho}U_0)^2}+\mathcal{O}(\varepsilon)\right).	
	\end{aligned}
	\]}
By \eqref{coordinates_change} and \eqref{equ_xi}, we can get
{\small
	\begin{equation}\label{div_m}
		\begin{aligned}
			&\nabla\cdot\mathbf{m}
			%=\varepsilon^{-1}\mathbf{\nu}\cdot\partial_{\rho}\mathbf{m}+\nabla_s\cdot\mathbf{m}+\mathcal{O}(\varepsilon)\\
			=\varepsilon^{-1}\partial_{\rho}(\mathbf{\nu}\cdot\mathbf{m})+\nabla_\mathbf{s}\cdot\mathbf{m}+\mathcal{O}(\varepsilon)\\
			=&\epsilon^{-1}\partial_{\rho}\left(\epsilon^{-1}\gamma_0\partial_{\rho}U_0+\left(\gamma_0\partial_{\rho}U_1+\gamma_1+\left(\mathbf{\nu}\cdot\nabla{\bm\xi}(\mathbf{\nu})\nabla_\mathbf{s}U_0-\gamma_1\right)(\frac{1}{2}+\frac{F(U_0)}{(\partial_{\rho}U_0)^2})	\right)	\right)\\
			&+\nabla_\mathbf{s}\cdot\left(\epsilon^{-1}\left(\gamma_0\partial_{\rho}U_0\mathbf{\nu}(\frac{1}{2}-\frac{F(U_0)}{(\partial_{\rho}U_0)^2})+\partial_{\rho}U_0{\bm\xi}(\mathbf{\nu})(\frac{1}{2}+\frac{F(U_0)}{(\partial_{\rho}U_0)^2})	\right)	\right)+\mathcal{O}(1).
		\end{aligned}
	\end{equation}
}

% Since $ \gamma_0=\hat{\gamma}(\mathbf{\nu}) $ is independent of $\rho$, 
%the leading order term of $ \nabla\cdot\mathbf{m} $ is $ \varepsilon^{-2}\gamma_0\partial_{\rho\rho}U_0 $. From the matching process at the leading order, we can obtain that $ U_0 $ is independent of $ s $, which implies $ \gamma_1=0 $. Moreover, we also have $ (\partial_{\rho}U_0)^2=2F(U_0) $. 
%We can simplify $ \nabla\cdot\mathbf{m} $ and obtain
%\begin{equation}\label{div_m2}
%	\nabla\cdot\mathbf{m}=\varepsilon^{-2}\gamma_0\partial_{\rho\rho}U_0+\varepsilon^{-1}(\gamma_0\partial_{\rho\rho}U_1+\nabla_s\cdot(\partial_{\rho}U_0{\bm\xi}(\mathbf{\nu})))+\mathcal{O}(1).
%\end{equation}
Combining \eqref{right_mu_first} and \eqref{div_m}, we finally obtain 
\eqref{EQ2}.
%the expansion of \eqref{eq2} as
% \begin{align*}
% \bar{\mu}=&\varepsilon^{-1}(\gamma_0F'(U_0)-\gamma_0\partial_{\rho\rho}U_0)+\left(\gamma_0F''(U_0)U_1+\gamma_1\frac{F'(U_0)}{\partial_{\rho}U_0}-\gamma_0\partial_{\rho\rho}U_1-\partial_{\rho}\gamma_1\right)\\
% &-\partial_{\rho}\left(\left(\mathbf{\nu}\cdot\nabla{\bm\xi}(\mathbf{\nu})\nabla_sU_0-\gamma_1\right)(\frac{1}{2}+\frac{F(U_0)}{(\partial_{\rho}U_0)^2})	\right)	\\
% &-\nabla_s\cdot\left(\gamma_0\partial_{\rho}U_0\mathbf{\nu}(\frac{1}{2}-\frac{F(U_0)}{(\partial_{\rho}U_0)^2})+\partial_{\rho}U_0{\bm\xi}(\mathbf{\nu})(\frac{1}{2}+\frac{F(U_0)}{(\partial_{\rho}U_0)^2})	\right)	+\mathcal{O}(\varepsilon).
% \end{align*}

\section{Numerical scheme}\label{numerical_scheme}
 In this paper, we adopt a linear numerical scheme based on the stabilized-invariant energy quadratization (S-IEQ) approach\cite{xu2019efficient,yang2021efficient} to verify the properties of the ACH-IC.
%To address the issue of ill-posedness in strongly anisotropic cases, we incorporate a Willmore regularized term, $\frac{\beta}{2}\left(\frac{1}{\varepsilon^2}F'(u)-\Delta u\right)^2$, into the energy functional. 
To facilitate the numerical approach, we introduce an auxiliary variable $V$ defined as:
\[
V=\sqrt{\frac{\gamma(\mathbf{\bhatn})}{\varepsilon}\left(F(u)+\frac{\varepsilon^2	}{2}|\nabla u|^2\right)+\frac{\beta}{2\varepsilon^3}\left(F'(u)-\varepsilon^2\Delta u\right)^2+B},
\]
%\[
%V=\sqrt{\gamma(\mathbf{n})\left(\frac{1}{\varepsilon}F(u)+\frac{\varepsilon}{2}|\nabla u|^2\right)+B},
%\]
where $B$ is a constant that ensures the expression's positivity under the square root. 
%For isotropic cases and weakly anisotropic cases, $B$ can be set to zero. However, for strongly anisotropic cases, it is necessary to choose $B\sim \mathcal{O}(\frac{1}{\varepsilon^2})$. With this auxiliary variable $ V $, the total energy can be represent as
With this new variable, the energy \eqref{Torabi_energy} become
\[
E(u,V)=\int_\Omega\left(V^2-B\right)\dx,
\] 
and we can reformulate the system of \eqref{eq_NMN} as:
\begin{equation}\label{new_ieq_system}
	\begin{aligned}
		&u_t=\frac{1}{\varepsilon}N(u) \nabla \cdot(M(u) \nabla (N(u)\mu)),\\
		&\mu=HV,\\
		&V_t=\frac{1}{2}Hu_t,
	\end{aligned}
\end{equation}
where 
\begin{align*}
H(u)&=\frac{\frac{1}{\varepsilon}\gamma(\mathbf{\bhatn})F'(u)-\varepsilon\nabla\cdot\mathbf{m}+\frac{\beta}{\varepsilon^2}\left(F''(u)\omega-\varepsilon^2\Delta \omega\right)}{\sqrt{\frac{\gamma(\mathbf{\bhatn})}{\varepsilon}\left(F(u)+\frac{\varepsilon^2	}{2}|\nabla u|^2\right)+\frac{\beta}{2\varepsilon^3}\left(F'(u)-\varepsilon^2\Delta u\right)^2+B}},\\
\mathbf{m}&=\gamma(\bhatn)\nabla u+\frac{(I-\bhatn\bhatn^T)\nabla_{\bhatn}\gamma(\bhatn)}{\sqrt{|\nabla u|^2+\varepsilon^2}}\left( \frac{1}{2}|\nabla u|^2+\frac{1}{\varepsilon^2}F(u)\right),\\
\omega&=\frac{1}{\varepsilon}F'(u)-\varepsilon\Delta u.
\end{align*}
%\[
%H(u)=\frac{\frac{1}{\varepsilon}\gamma(\mathbf{n})F'(u)-\varepsilon\nabla\cdot\mathbf{m}}{\sqrt{\gamma(\mathbf{n})\left(\frac{1}{\varepsilon}F(u)+\frac{\varepsilon}{2}|\nabla u|^2\right)+B}}.
%\]
%Here, we modify the definition of $\mathbf{m}$ in \eqref{eq_NMN} using the second asymptotic approximation $ F(u)=\frac{\varepsilon^2}{2}|\nabla u|^2 $ which is a commonly used method \cite{torabi2009new, xu2019efficient} to address the singularity in numerical simulations arising from the term $ \frac{F(u)}{|\nabla u|} $. The modified expression for $\mathbf{m}(u)$ is given by:
%\[
%\mathbf{m}(u)=\gamma(\mathbf{n})\nabla u+|\nabla u|(I-\mathbf{n}\mathbf{n}^T)\nabla_{\mathbf{n}}\gamma(\mathbf{n}).
%\]
We consider a time step size $\delta t>0$ and define $ t_n=n\delta t $ for $ 0\leqslant n\leqslant N $ with final time $T=N\delta t$. With these definitions, we use a semi-discrete time discretization scheme based on the Invariant Energy Quadratization (IEQ) approach to solve the equivalent system \eqref{new_ieq_system}:
{\small
	\begin{numcases}{}
		\frac{u^{n+1}-u^n}{\delta t}=\frac{1}{\varepsilon}N(u^n) \nabla \cdot(M(u^n) \nabla (N(u^n)\mu^{n+1})),\label{ieq_system1_1}\\
		\mu^{n+1}=H^nV^{n+1}+\frac{S_1}{\varepsilon}(u^{n+1}-u^n)-S_2\varepsilon\Delta(u^{n+1}-u^n)+S_3\varepsilon\Delta^2(u^{n+1}-u^n),\label{ieq_system1_2}\\
		V^{n+1}-V^{n}=\frac{1}{2}H^n(u^{n+1}-u^n)\label{ieq_system1_3}.
	\end{numcases}
}where  $ S_1 $, $ S_2 $, $ S_3 $ are three stabilization constants, and $ H^n=H(u^n) $. It should be remarked that using the Adams-Bashforth backward differentiation formula, we can also construct a second-order time-stepping numerical scheme. The investigation of efficient and sophisticated numerical schemes is beyond the scope of this work.
\begin{remark}
	In order to enhance energy stability and allow for larger time step sizes, we incorporate two additional stabilization terms into the numerical scheme. These stabilization terms help mitigate non-physical oscillations caused by the anisotropic coefficient $\gamma(\mathbf{\hat{n}})$, as the normal vector undergoes frequent sign changes in bulk regions where $u\approx\pm1$. Although unconditional energy stability can be theoretically shown without these stabilization terms, the inclusion of such terms helps avoid non-physical oscillations in numerical simulations. Similar ideas and numerical results can be found in Refs.~\cite{xu2019efficient,yang2021efficient}.
\end{remark}	
\begin{remark} 	
	To balance the explicit treatment of $ \frac{1}{\varepsilon}\gamma(\mathbf{n})F'(u) $, we introduce one term $ \frac{S_1}{\varepsilon}(u^{n+1}-u^n) $. This term introduces an additional consistency error of order $ \frac{S_1\delta t}{\varepsilon}u_t $, which is of the same order as the error introduced by the explicit treatment of $ \frac{1}{\varepsilon}\gamma(\mathbf{n})F'(u) $. Similarly, the other stabilization term is introduced to balance $ -\nabla\cdot(\gamma(\mathbf{n})\nabla u) $ and $\Delta \omega$.  
	%	By introducing these stabilized terms, we obtain a structure resembling a convex-concave structure. Based on the idea of convex splitting methods, we conjecture that the orders of $ S_1 $, $ S_2 $, and $ S_3 $ are $\mathcal{O}(1)$. 
\end{remark}

The energy stability can be proved based on the standard arguments of the IEQ approach which is shown as follows:
\begin{theorem}[energy stable]
	The numerical scheme \eqref{ieq_system1_1}-\eqref{ieq_system1_3} is unconditional energy stable and the following discrete energy law holds for
	any $\delta t>0$:
	\begin{equation}\label{energy_law}
		\frac{1}{\delta t}(E^{n+1}-E^n)\leqslant-\frac{1}{\varepsilon}\int_\Omega M(u^n)|\nabla(N(u^n)\mu^{n+1})|^2\dx\leqslant0,
	\end{equation}
	where $ E^n=\int_\Omega |V^{n}|^2\dx-B|\Omega| $.
\end{theorem}
\begin{proof}
	We first take $ L^2 $ inner product of \eqref{ieq_system1_1} with $ \mu^{n+1} $ in $\Omega$. Using the integration by parts and the periodic boundary condition, we have 
	\begin{equation}\label{energy_law_eq1}
		\frac{1}{\delta t}\int_\Omega\left(u^{n+1}-u^n\right)\mu^{n+1}\dx=-\frac{1}{\varepsilon}\int_\Omega M(u^n)|\nabla(N(u^n)\mu^{n+1})|^2\dx
	\end{equation}
	Similarly, by taking $ L^2 $ inner product of \eqref{ieq_system1_2} with $ (u^{n+1}-u^n) $ in $\Omega$, we obtain
	\begin{equation}\label{energy_law_eq2}
		\begin{aligned}
			\int_\Omega(u^{n+1}-u^n)\mu^{n+1}\dx&=\int_\Omega H^nV^{n+1}(u^{n+1}-u^n)\dx+\frac{S_1}{\varepsilon}\int_\Omega|u^{n+1}-u^n|^2\dx\\
			&~~~+S_2\varepsilon\int_\Omega|\nabla(u^{n+1}-u^n)|^2\dx+ S_3\varepsilon\int_\Omega|\Delta(u^{n+1}-u^n)|^2\dx\\
			&\geqslant\int_\Omega H^nV^{n+1}(u^{n+1}-u^n)\dx.
		\end{aligned}
	\end{equation}
	At last, we take $ L^2 $ inner product of \eqref{ieq_system1_3} with $ 2V^{n+1} $ in $\Omega$:
	\begin{equation}\label{energy_law_eq3}
		\begin{aligned}
			\int_\Omega H^nV^{n+1}(u^{n+1}-u^n)\dx&=\int_\Omega2V^{n+1}(V^{n+1}-V^n)\dx\\
			&=\int_\Omega |V^{n+1}|^2\dx-\int_\Omega |V^{n}|^2\dx+\int_\Omega |V^{n+1}-V^n|^2\dx\\
			&\geqslant E^{n+1}-E^n,
		\end{aligned}
	\end{equation}
	where the second equality is due to the identity $ 2(a-b)a=a^2-b^2+(a-b)^2 $.
	
	Combing \eqref{energy_law_eq1}, \eqref{energy_law_eq2} and \eqref{energy_law_eq3}, we get the discrete energy law \eqref{energy_law}. 
\end{proof}

\begin{remark}
	Our numerical method guarantees the decay of the modified energy; however, it does not explicitly enforce the boundedness of the phase variable \( u \) within the physical range \([-1, 1]\). This limitation could lead to non-physical solutions in certain scenarios, which is undesirable in practice. Addressing this issue represents an important direction for future research to improve the robustness and physical consistency of the method.
\end{remark}

To discretize the equations \eqref{ieq_system1_1}-\eqref{ieq_system1_3} in space, we employ finite differences on a uniformly spaced grid. Unless otherwise specified, we progressively refine the mesh until the discretization error becomes negligible. For instance, in the case of \(\varepsilon=0.01\), \(\delta x=1/256\) is sufficient to show the convergence order. The stabilization parameters in the numerical scheme  \eqref{ieq_system1_1}-\eqref{ieq_system1_3} are $S_1=S_2=4$, $S_3=0$, consistent with the values used in Ref.~\cite{xu2019efficient}. The resulting linear systems are solved using the generalized minimal residual method (GMRES) with an incomplete LU preconditioner. In practice, we utilized the built-in MATLAB solver for the efficient solution of these linear equations. 

\section*{Acknowledgments}
The work of Zhen Zhang was partially supported by National Key R\&D Program of China (2023YFA1011403), the NSFC grant (92470112), the Shenzhen Sci-Tech Inno-Commission Fund (20231120102244002), and the Guangdong Provincial Key Laboratory of Computational Science and Material Design (No. 2019B030301001). The work of Wei Jiang was partially supported by the NSFC Grant (No.~12271414).
The work of Tiezheng Qian was partially supported by the Hong Kong RGC General Research Fund (No.~16306121) and the Key Project of the National Natural Science Foundation of China (No.~12131010).

%%%% Insert A head here

%%%% 

%%%%%%%%%% Insert bibliography 
\bibliographystyle{plain}
\bibliography{refs}

\end{document}